\documentclass[12pt, draftclsnofoot,onecolumn]{IEEEtran}
\usepackage{textcomp}
\usepackage{cite}
\usepackage{graphicx}
\usepackage{balance}
\usepackage{color}
\usepackage{url}
\usepackage{graphicx, pstool}
\usepackage{float}
\usepackage{array}
\usepackage[english]{babel} % English language/hyphenation
\usepackage{amsmath,amssymb,amsfonts}
\usepackage{textcomp}
\usepackage{algorithmic}
\usepackage{caption}
\usepackage{subcaption}
\usepackage{breqn}
\usepackage[ruled,vlined,linesnumbered]{algorithm2e}
\usepackage{cuted}
\usepackage{amsthm}

\normalsize
\makeatletter
\renewcommand{\fnum@figure}{Fig. \thefigure}
\makeatother

\newcolumntype{P}[1]{>{\centering\arraybackslash}p{#1}}
\newtheorem{theorem}{Theorem}

\newtheorem{lemma}{Lemma}

\newtheorem{corollary}{Corollary}

\newtheorem{proposition}{Proposition}

\newtheorem{conjecture}{Conjecture}

\newtheorem{sketch}{Sketch of Proof}

\def\BibTeX{{\rm B\kern-.05em{\sc i\kern-.025em b}\kern-.08em
		T\kern-.1667em\lower.7ex\hbox{E}\kern-.125emX}}

\begin{document}
	
	%-------------------- paper title----------------------%
	\title{Performance Analysis of Relay Selection Schemes in Multi-Hop Decode-and-Forward Networks
		%	\thanks{This work was supported by the Australian Research Council Discovery Project under Grant DP180101205 and Discovery Early Career Researcher Award under Grant DE180100501.}
	}
	
	\author{
		Shalanika Dayarathna,~\IEEEmembership{Member,~IEEE,} Rajitha Senanayake,~\IEEEmembership{Member,~IEEE,} \\ and Jamie Evans,~\IEEEmembership{Senior Member,~IEEE}
	}
	\maketitle
	
	%-------------------------- Summary-------------------%
	\begin{abstract}		
	This paper analyses the data rate achieved by various relay selection schemes in a single-user multi-hop relay network with decode-and-forward (DF) relaying. While the single-user relay selection problem is well studied in the literature, research on achievable rate maximization is limited to dual-hop networks and multi-hop networks with a single relay per hop. We fill this important gap by focusing on achievable rate maximization in multi-hop, multi-relay networks. First, we consider optimal relay selection and obtain two approximations to the achievable rate. Next, we consider three existing sub-optimal relay selection strategies namely hop-by-hop, ad-hoc and block-by-block relay selection and obtain exact expressions for the achievable rate under each of these strategies. We also extend the sliding window based relay selection to the DF relay network and derive an approximation to the achievable rate. Further, we investigate the impact of window size in sliding window based relay selection and show that a window size of three is sufficient to achieve most of the possible performance gains. Finally, we extend this analysis to a noise limited multi-user network where the number of available relay nodes is large compared to the number of users and derive approximations to the achievable sum-rate.
	\end{abstract}
	
	\begin{IEEEkeywords}
		Achievable rate, ad-hoc, block-by-block, decode-and-forward, hop-by-hop, multi-hop, optimal relay selection, single-user, sliding window.	
	\end{IEEEkeywords}
	
	%------------------ Introduction ------------------%
	\section{Introduction} \label{section-intro}
	With current 5G and upcoming 6G technologies, the number of connected devices is ever increasing. As a result, future wireless networks will need to enhance the spectral efficiency and cooperative communication has been introduced as a key strategy to achieve spatial diversity and to combat fading \cite{00047}. In cooperative networks, multiple nodes known as relay nodes are placed across the coverage area to support the communication between the transmitter and the intended receiver. For such networks, that typically arise in sensor mesh and ad hoc communications, cellular communication and vehicular communication, a proper choice of relays can improve the performance while mitigating the overhead and complexity. Therefore, the relay selection in multi-hop cooperative relay networks has resulted in a significant research interest \cite{my_Thesis,2970744,s18103263,3177187}.
	
	In the relay network literature, the main focus of the single-user, multi-hop relay selection problem has been on the objective of maximizing the minimum signal-to-noise-ratio (SNR) or minimizing the outage probability. In \cite{2008.107}, the authors propose a simple hop-by-hop relay selection strategy for a single-user, multi-hop decode-and-forward (DF) relay network where the per hop relay selection is optimized based on the local hop gains. This was later extended in \cite{071030}, to achieve the full-diversity where the authors proposed an ad-hoc relay selection strategy by joining the final two hops together. As the performance loss of ad-hoc routing increases with the number of hops, they also propose a block-by-block relay selection strategy by considering a block of non-overlapping hops at a time to improve the performance at the expense of complexity. Taking a different approach, in \cite{5982498}, a new framework for optimal route selection in single-user, multi-hop DF relay networks is proposed by mapping the relay network into a trellis diagram. It was shown that the dynamic programming based Viterbi algorithm can be used to find the optimal relay assignment by finding the most likely sequence of states in the trellis. However, an direct extension of the dynamic programming is not feasible for amplify-and-forward (AF) relay networks. As such, the Viterbi algorithm based near optimal routing scheme is proposed in \cite{6364160}, where the effective SNR for a specific relay path is approximated by the minimum value of SNRs across all the hops, which results in an approximation that is tight at high SNR. Following a similar approach, a dynamic programming based algorithm for route selection is proposed in \cite{130145}, where the authors use a recursive computation of the outage probability. In \cite{152973}, the authors derive analytical expressions for outage probability under Nakagami fading with consideration to diversity order for a DF multi-hop relay network with single relay node in each hop. A similar analysis is conducted in \cite{en11113004}, where the authors study a multi-hop parallel relay network under Nakagami fading. As a result, the average end-to-end rate throughput is calculated either using the required per-hop-rate, the number of hops and the end-to-end outage probability \cite{2969232} or using the allocated bandwidth and the spectral efficiency \cite{0286-3}. 
	
	With the increasing demand for high data rates, the throughput maximization has gained more attention in the recent literature \cite{2684125}. In terms of throughput maximization, most research has focused on the simple dual-hop relay networks where both hops can be easily combined together to compute the end-to-end achievable rate \cite{122813,2360174,6162468,9771581}. In \cite{122813}, Meijer's function is used to derive an analytical expression for the achievable rate of a dual-hop relay network with a single relay node under the Rician fading distribution. This was later extended to a multiple antenna dual-hop relay network in \cite{2360174} and to a relay network with multiple relay nodes in \cite{6162468}. In \cite{6162468}, the authors use Marcum Q-function to analyze the ergodic capacity of a single-user, dual hop relay network with multiple relay nodes under the Rician fading environment. Taking a different approach, in \cite{071098,6811374,7248654}, the authors analyze the capacity of a single-user, dual-hop relay network with multiple relay nodes under the Rayleigh fading distribution. In a multi-hop relay network, the relay selection in each hop determines the available relay choices in the next hop. This introduces an added complexity for a multi-hop network compared to a dual-hop network. The capacity of a multi-hop relay network with a single relay node in each hop is considered in \cite{874520,262121} where the authors analyze the capacity in terms of the energy efficiency and the spectral efficiency. In \cite{6825593}, the authors obtain an analytical expression for the capacity of a single-user, multi-hop DF relay network with a single relay in each hop under the Rician fading distribution. Similar analysis has been performed in \cite{080818} for a DF relay network under the Rayleigh fading channel and in \cite{080118} for an AF relay network under the adaptive transmission with Nakagami-m fading distribution. However, with single relay node in each hop, there would be no relay selection involved, thus making the rate analysis in these work significantly different to the analysis considered in this paper. As a result, even though, both maximizing the minimum SNR and maximizing the overall achievable rate provide same relay assignment for the special case of single user network, the analytical expressions with respect to the achievable rate have not been analyzed for a multi-hop DF relay network consisting of multiple relay nodes in each hop. To the best of our knowledge, there exists no analytical expressions to find the achievable rate in a single-user, multi-hop, multi-relay DF relay network. We make progress towards addressing this challenge by analyzing the achievable rate of a multi-hop DF relay network with multiple relay nodes under the Rayleigh fading distribution. More recent work in the field of relay selection in single-user, multi-hop relay networks focuses on machine learning and deep reinforcement learning techniques \cite{electronics8090949,e23101310,2941932}. However, these work focus on providing more practical relay selection solutions while in this paper, we focus on theoretical analysis of the achievable rate.
	
	In this paper, we consider a single-user, multi-hop DF relay network where each hop consists of multiple relays and analyze the relay selection problem that aims to maximize the achievable rate throughput. In the following, we list the contributions of this paper. 
	\begin{itemize}
		\item By considering the optimal relay selection strategy, we derive two approximations to the optimal achievable rate as presented in \eqref{optimal_approx1} and \eqref{optimal_approx2}. Further, we show that the second approximation, which is derived based on dynamic programming, is more accurate than the simple approximation that can be derived using the independent path model. 
		\item By considering three sub-optimal relay selection strategies namely, hop-by-hop relay selection, ad-hoc relay selection and block-by-block relay selection, we derive the exact expressions for the achievable rate. This contribution is presented in Theorem~\ref{theorem_hop}, Theorem~\ref{theorem_adhoc} and Theorem \ref{theorem_block}, respectively. 
		\item Next, we consider sliding window based relay selection strategy used in AF relay networks. Extending this to the DF relay networks, we derive an approximation to the achievable rate under sliding window based relay selection. This contribution is presented in \eqref{eq_sliding_rate}. Based on simulations, we further show that under this strategy, the window size of three is sufficient to achieve most of the possible performance gains compared to optimal relay selection.
		\item Finally, we extend this rate analysis to a noise limited multi-user network where there is a larger number of relay nodes per hop compared to the the number of users. For such network, we show that the achievable sum-rate can be approximated by the single user rate times the number of users. This contribution is presented in Section \ref{section-multi}.	
	\end{itemize} 	
	The structure of the paper is along the following lines. System model followed by the formulation of the optimization problem is provided in Section \ref{section-model} for a single-user, multi-hop DF relay network with multiple relay nodes in each hop. Next, the optimal relay selection strategy and four sub-optimal relay selection strategies are analyzed in terms of the achievable rate in Section \ref{section-single} with extension to multi-user network is given in Section \ref{section-multi}. Finally, the conclusions and potential future work is outlined in Section \ref{section-conclusion}.
	
	\section{System Model and Optimization Problem Formulation}\label{section-model}	
	\begin{figure}
		\centerline{\includegraphics[width=0.75\textwidth]{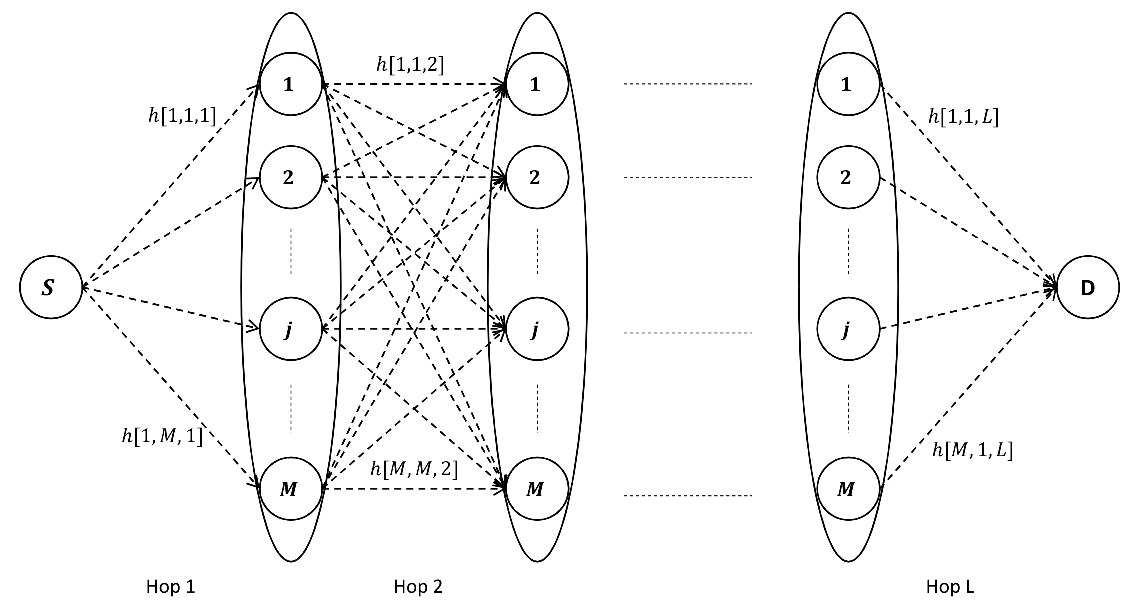}}
		\caption{A single-user, multi-hop relay network.}
		\label{figure1}
	\end{figure}	
	We consider a single-user relay network as illustrated in Fig. \ref{figure1}, where the source node $S$ sends information to the destination node $D$. The communication between source and destination is facilitated via a wireless relay network consists of $L$ hops where each hop contains $M$ DF relay nodes. Please note that the more general case with different number of relay nodes in each hop can be handled by considering dummy relay nodes in each hop \cite{axiv_paper1}. As such, the assumption of equal number of relay nodes in each hop is made without loss of generality. With the intention of minimizing the required conditions for synchronization among relay nodes, we consider that the communication in each hop is facilitated by only a single relay node and denote the relay selected at hop $l$ as $r(l)$. We also assume that each relay node operates in half-duplex mode with maximum transmit power $P$ for transmission and that transmissions are scheduled such that cross-hop interference can be neglected \cite{2008.107,071030,5982498,6364160,6825593}. 
	As commonly used in the literature, we consider all links to be independent but non-identically distributed. In this work, we limit our analysis to Rayleigh fading distribution as commonly considered in the literature \cite{071098,6811374,7248654,080818}. As such, the Rayleigh fading channel coefficient of the link in hop $l$ from transmitting node $i$ to receiving node $j$ can be modeled as a random variable denoted by $h[i,j,l]$ with mean zero and variance $\sigma_1^2$, where $\sigma_1^2$ results from the large scale fading \cite{7636773}.  
	
	When node $i$ is transmitting in hop $l$, we can write the received signal at node $j$ as,
	\begin{align}
		y[j,l]=\sqrt{P}\,h[i,j,l]\,s[i,l-1]+n[j,l]
	\end{align}
	where $s[i,l-1]$ and $n[j,l]$ denote the transmitted data symbol by the transmitting node $i$ and the additive white Gaussian noise (AWGN) at the receiving node $j$ in hop $l$, respectively. The average energy of the data symbols are normalized to one and the AWGN distribution is with mean zero and variance $\sigma^2$. Noting that the corresponding end-to-end received SNR is determined by the minimum SNR value across all $L$ hops, we write this as,
	\begin{align}\label{eq_1}
		\gamma^{min} = \dfrac{P}{\sigma^2}\; {\underset{l \in \{1,...,L\}} {\textrm{min} }}\; \biggl\{|h[r(l-1),r(l),l]|^2\biggr\},
	\end{align}
	where $r(0)\!=\! r(L)\!=\!1$ since the signal is transmitted from $S$ in hop $0$ and received at $D$ in hop $L$. Therefore, we write the achievable rate maximization problem based on relay selection as,
	\begin{align}\label{eq_2}
		& {\underset{r(1),...,r(L-1)} {\textrm{max} }}\; \log_2(1+ \gamma^{min}) \nonumber\\
		&{\rm{s.t \;\; \;}}
		r(l) \in \{1,2,...,M\} ~~ \forall l, 
	\end{align}
	where $l \in \{1,2,..,L-1\}$. Due to the availability of multiple hops and the dependency of relay selection in adjacent hops, the challenge here is selecting the best relay combination such that the end-to-end achievable rate is maximized. We note that due to the monotonically increasing nature of the logarithmic function, the objective function in \eqref{eq_2} can be maximized by maximizing $\gamma^{min}$. As such, we consider a number of relay selection strategies used in the literature to maximize the end-to-end received SNR with the objective of minimizing the outage probability.
	
	\section{Relay Selection Strategies}\label{section-single}
	In this section, we consider the optimal relay selection and four sub-optimal relay selection strategies namely, hop-by-hop relay selection \cite{2008.107}, ad-hoc relay selection \cite{071030}, block-by-block relay selection \cite{071030} and sliding window based relay selection \cite{6364160} and derive novel analytical expressions for the achievable rate under each of these relay selection strategies. Please note that the work in \cite{2008.107,071030} and \cite{6364160} is limited to outage probability analysis. As such, the similarity between those work with current analysis is limited to the definitions of the sub-optimal relay selection strategies and the analytical expressions derived in this section is significantly different to existing work.
	
	\subsection{Optimal Relay Selection}
	In a relay network consisting of $L$ hops and $M$ relays at each hop, there are $M^{L-1}$ possible paths where each path can be viewed as a sequence of $L-1$ integers, corresponding to the relay indices, that take value between $1$ and $M$. Therefore, a path represents a sequence of relay nodes that facilitates the communication between source and destination nodes over $L$ hops. As such, optimal relay selection would involve selecting a path such that the end-to-end received SNR is maximized. When two paths retain the relay nodes in $l$ hops while allowing the relay nodes in remaining $L-l$ hops to be different, those two paths undergo same channels in $l$ hops. This makes them not independent to each other. Therefore, all the paths are not independent and as a result, obtaining an exact expression for the achievable rate is challenging. Thus, we take a more tractable approach by considering an approximate relay network with independent paths.
	
	\subsubsection*{Independent Path Model}
	In the following, we consider an approximated model based on the assumption that all $M^{L-1}$ paths are independent. Next, we derive an exact expression for the optimal achievable rate under this independent path model as presented in Lemma \ref{Lemma1}.
	\begin{lemma}\label{Lemma1}
		Consider a single-user, multi-hop relay network with $L$ hops and $M$ relay nodes per hop. Under the independent path assumption described above, the optimal achievable rate can be expressed as,
		\begin{align*}
			R^{\textrm{opt}}_{\textrm{ind}} = \dfrac{1}{\log 2}\sum_{q=1}^{Q} \binom{Q}{q}(-1)^{q} \; e^{Lq/\sigma_a^2}\; Ei\biggl(-\dfrac{Lq}{\sigma_a^2}\biggr),
		\end{align*}	
		where $\sigma_a^2=\frac{P\sigma_1^2}{\sigma^2}, Q=M^{L-1}$, $\binom{Q}{q}$ is the binomial coefficient and $Ei(.)$ denotes the exponential integral function \cite{2014520}.  
	\end{lemma}
	\begin{proof}
		Please refer to Appendix \ref{app7:2}.
	\end{proof}
	\noindent As such, we can approximate the optimal achievable rate using Lemma \ref{Lemma1} as,
	\begin{align}\label{optimal_approx1}
		R^{\textrm{opt}} \approx \dfrac{1}{\log 2}\sum_{q=1}^{Q} \binom{Q}{q}(-1)^{q} \; e^{Lq/\sigma_a^2}\; Ei\biggl(-\dfrac{Lq}{\sigma_a^2}\biggr).
	\end{align}
	In the following example, we analyze the accuracy of the approximation given in \eqref{optimal_approx1}.
	
	\vspace{0.25cm}
	\noindent
	\textbf{Example 1}: Consider a single-user, multi-hop relay network where $M=2, L=3,6$. Assuming constant distance in all hops resulted from the uniform placement of nodes \cite{2969232}, we normalize the variance of the channel fading coefficient, $\sigma_1^2$, to unity. For such a network, the analytical and the simulated achievable rate against the received SNR are plotted in Fig. \ref{figure2} under the optimal relay selection strategy. From the figure, we observe that the approximate analytical rate achieved with \eqref{optimal_approx1} is accurate at small $L$. However, with increasing $L$, the approximation given in \eqref{optimal_approx1} becomes loose. In addition, we can observe that the achievable rate decreases with increasing $L$. Since, it is assumed that the large scale fading is constant for each link, the fading of the end-to-end received SNR increases with $L$, thereby decreasing the achievable rate. 
	\begin{figure}
		\centerline{\includegraphics[width=0.7\textwidth]{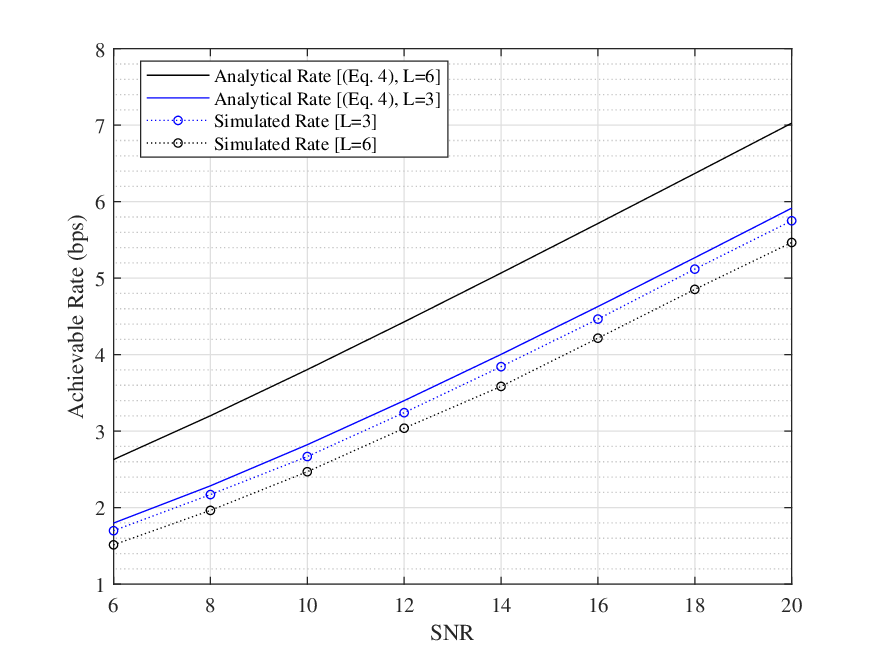}}
		\caption{The achievable rate against the received SNR under optimal relay assignment}
		\label{figure2}
	\end{figure}	
	
	Out of $M^{L-1}$ total paths there can be only $M$ paths that are independent in all $L$ hops and $M^{L-l-1}$ paths that are independent in at most $l$ hops with $l<L-2$. As such, number of dependent paths increases with $L$. This reduces the feasibility of independent paths assumption, thus making the approximation given in \eqref{optimal_approx1} not accurate with larger $L$ as illustrated in Example~1. Therefore, taking a different approach, we derive a much tighter approximation on the optimal achievable rate considering the dynamic programming based implementation proposed in \cite{5982498}. We note that the optimal relay path can also be obtained following a dynamic programming approach where a central controller can be used to select the path that has the largest bottleneck-link SNR with a linear complexity of $L$. By modeling the multi-hop relay network as a trellis diagram, the optimal relay path can be obtained via the one that results in the maximum value of the minimum branch weight across the stages. As a result, at state $i$ in stage $l$ we can write the Bellman equation relating to the cost function as,
	\begin{align}\label{eq_dynamic}
		V(i,l) = {\underset{j \in \{1,...,M\}} {\textrm{max}}}\biggl(\textrm{min}\biggl\{\dfrac{P}{\sigma^2} |h[j,i,l]|^2, V(j,l-1)\biggr\}\biggr),
	\end{align}
	where each stage is mapped to each hop in the network and each state is mapped to each relay node in a given hop. For further details on how to apply dynamic programming please refer to \cite{5982498}. We note that the initialization value in the first hop, $V(i,1)$, takes the value of the received SNR of the source-relay link for a given relay node $i$. For any other hop, $V(i,l)$ can be computed using \eqref{eq_dynamic}, and $V(1,L)$ represents the optimal end-to-end received SNR. As such, each $V(i,l)$ depends on the minimum SNR up to that point and they are not independent of each other. This makes obtaining an exact expression for the achievable rate challenging. Therefore, we assume that $V(i,l)$ for a given $l$ is independent of every other $V(\bar{i},l), \forall i\neq\bar{i}$ and derive a second approximation on the optimal achievable rate as follows.
	\begin{align}\label{optimal_approx2}
		R^{\textrm{opt}} \approx \dfrac{1}{\log 2}\sum_{q=1}^{\tilde{Q}} a_q \; e^{b_q/\sigma_a^2}\; Ei\biggl(-\dfrac{b_q}{\sigma_a^2}\biggr),
	\end{align}	
	where $\sigma_a^2=\frac{P\sigma_1^2}{\sigma^2}, \tilde{Q}=2M^{L-1}+\sum_{l=1}^{L-2}M^l, \; \forall L>1$ and the coefficients $a_q$ and $b_q$ can be numerically computed for a given $M$ and $L$ by solving the recursive function
	\begin{align*}
		\sum_{q=0}^{\tilde{Q}} a_q \; e^{-b_q\,x/\sigma_a^2} &= F_{V(1,L)}(x) = \bigg(1-e^{-x/\sigma_a^2}\bigg(1-F_{V(1,L-1)}(x)\bigg)\bigg)^M,
	\end{align*}
	with $F_{V(1,1)}(x)=1-e^{-x/\sigma_a^2}$.
	
	\noindent In the following example, we compare the two approximations given in \eqref{optimal_approx1} and \eqref{optimal_approx2}.
	
	\vspace{0.25cm}
	\noindent
	\textbf{Example 2}: Consider a single-user, multi-hop relay network similar to Example 1. For such a network, the analytical and the simulated achievable rate against the received SNR are plotted in Fig. \ref{figure7} and Fig. \ref{figure8} under the optimal relay selection strategy. From the Fig. \ref{figure7}, we observe that the approximate analytical rate achieved with \eqref{optimal_approx2} is more accurate than \eqref{optimal_approx1}. We can also observe that with increasing $M$ both approximations become tight. On the other hand, from the Fig. \ref{figure8}, we observe that with increasing $L$, both approximations become loose while the analytical approximation achieved with \eqref{optimal_approx2} is much more accurate than that of \eqref{optimal_approx1} when $L$ is large. The computation of \eqref{optimal_approx1} involves summation $Q$ numbers and computing $\binom{Q}{q}$ with a time complexity $O(Q)$ where $Q=M^{L-1}$. As such, the approximation in \eqref{optimal_approx1} has a time complexity $O(M^{L-1})$. The computation of \eqref{optimal_approx2} involves summation of $\bar{Q}$ numbers, finding $\bar{Q}+1$ coefficients by comparing two polynomials and computing a recursive function of time complexity $O(M)$ for $L$ times where $\bar{Q}=2M^{L-1}$. As such, the approximation in \eqref{optimal_approx2} has a time complexity $O(M^{L-1})+O(ML)$. Therefore, for smaller $L$ values, the approximation in \eqref{optimal_approx1} has a less computational complexity while providing similar result to \eqref{optimal_approx2}. On the other hand, for larger $L$ values, the complexity of both approximations become closer with the approximation in \eqref{optimal_approx2} providing more accurate results.
	\begin{figure}
		\begin{subfigure}{.5\textwidth}
				\centering
				\includegraphics[width=\textwidth]{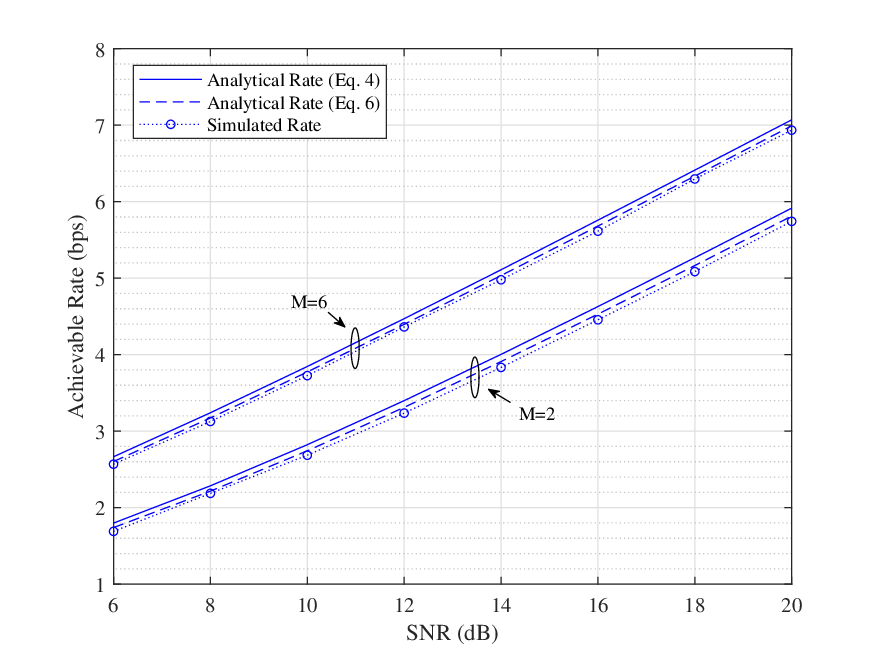}
				\captionsetup{justification=centering}
				\caption{$M=\{2,6\}, L=3$}
				\label{figure7}
			\end{subfigure}
		\begin{subfigure}{.5\textwidth}
				\centering
				\includegraphics[width=\textwidth]{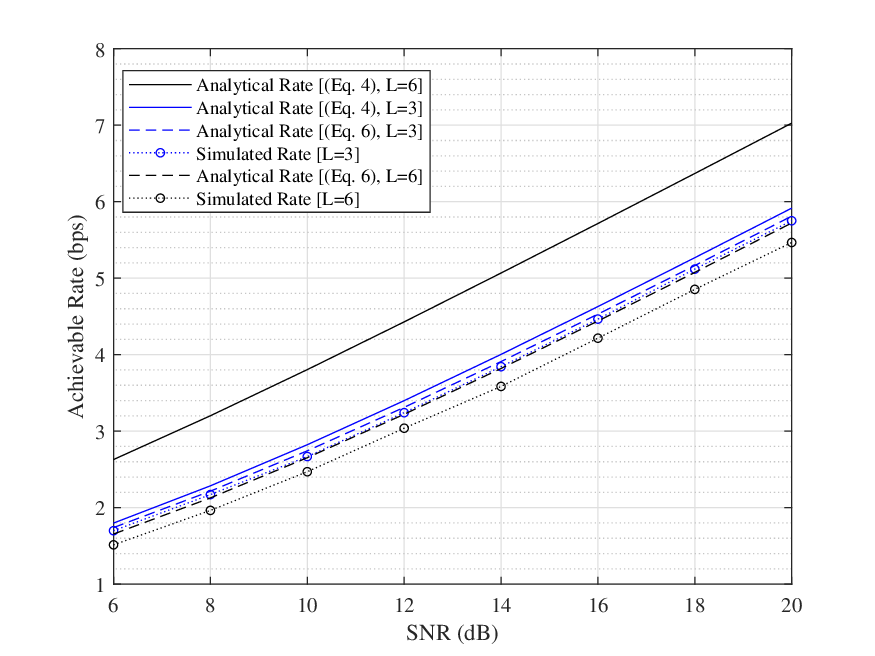}
				\captionsetup{justification=centering}
				\caption{$M=2, L=\{3,6\}$}
				\label{figure8}
			\end{subfigure}
		\caption{The achievable rate against the received SNR}
	\end{figure}	
	
	Even though the optimal relay assignment that maximizes the achievable rate can be obtained with a linear complexity in respect to $L$ and $M$, it requires the global channel state information of all the links in the network. As a result, we next consider, four sub-optimal relay selection strategies that can be implemented in a distributed environment. 
	
	\subsection{Hop-by-Hop Relay Selection}
	As the first sub-optimal relay selection strategy, we consider the simple hop-by-hop relay selection. Under this strategy, the relay selection in each hop is performed independently such that the achievable rate of the current hop is maximized when the signal is transmitted from the relay selected in the previous hop. Therefore, with this strategy, there is no rate optimization involved in the last hop where the destination node is fixed.
	\begin{theorem}\label{theorem_hop}
		Consider a single-user, multi-hop relay network with $L$ hops and $M$ relay nodes per hop. For such a network, the achievable rate under hop-by-hop relay selection is given by 
		\begin{align*}%\label{eq_hop_rate}
			R^{\textrm{hop}} = -\dfrac{1}{\log 2}\sum_{{\underset{= L-1}{l_1+...+{l_M}}}} \!\! A^{L-1}_{l_1...l_M} \; e^{\beta^{L-1}_{l_1...l_M}/\sigma_a^2}\; Ei\biggl(-\dfrac{\beta^{L-1}_{l_1...l_M}}{\sigma_a^2}\biggr),
		\end{align*}	
		where $A^{L-1}_{l_1...l_M} = \binom{L-1}{l_1,...,l_M}\biggl(\prod_{j=1}^{M}\binom{M}{j}^{l_j}\biggr)(-1)^{\sum_{j=1}^{M}(j-1)l_j}$, $\beta^{L-1}_{l_1...l_M} = 1+\sum_{j=1}^{M}jl_j, \sigma_a^2=\frac{P\sigma_1^2}{\sigma^2},$ with $l_1,..,l_M$ denote the multinomial expansion exponents, $\binom{L-1}{l_1,...,l_M}$ and $\binom{M}{j}$ denote the multinomial and the binomial coefficients, respectively, and $Ei(.)$ is the exponential integral function.  
	\end{theorem}
	\begin{proof} 
		The CDF of $\gamma^{min}$ under hop-by-hop relay selection can be expressed as,
		\begin{align}\label{CDF_hop}
			F_{\gamma^{min}}^{\textrm{hop}}(x) 
			&\!=\! 1\!-\!e^{(-x/\sigma_a^2)}\bigg(1\!-\!\bigg(1\!-\!e^{(-x/\sigma_a^2)}\bigg)^M\bigg)^{L\!-\!1}.
		\end{align}
		Please refer to Appendix \ref{app7:3} for a detailed derivation of \eqref{CDF_hop}. Next, we use the multinomial expansion \cite{00102} and re-write \eqref{CDF_hop} as,
		\begin{align}
			F_{\gamma^{min}}^{\textrm{hop}}(x) &= 1-\sum_{{\underset{= L-1}{l_1+...+{l_M}}}} \bigg[ \binom{L-1}{l_1,...,l_M}\biggl(\prod_{j=1}^{M}\binom{M}{j}^{l_j}\biggr)
			(-1)^{\sum_{j=1}^{M}(j-1)l_j} e^{-(1+\sum_{j=1}^{M}jl_j)x/\sigma_a^2}\bigg].
		\end{align}
		Taking the derivative with respect to $x$, we can derive the PDF of $\gamma^{min}$ as,
		\begin{align}
			f_{\gamma^{min}}^{\textrm{hop}}(x) = \!\! \sum_{{\underset{= L-1}{l_1+...+{l_M}}}} \!\! A^{L-1}_{l_1...l_M}\biggl(\dfrac{\beta^{L-1}_{l_1...l_M}}{\sigma_a^2}\biggr)e^{-\beta^{L-1}_{l_1...l_M}x/\sigma_a^2}.
		\end{align}
		As such the achievable rate can be derived using the variable transformation $t = 1+x$ and \cite[eq. (4.331.2)]{2014520} as,
		\begin{align}
			R^{\textrm{hop}} &= \dfrac{-1}{\log 2}\sum_{{\underset{= L-1}{l_1+...+{l_M}}}} \!\! A^{L-1}_{l_1...l_M}\; e^{\beta^{L-1}_{l_1...l_M}/\sigma_a^2}\; Ei\biggl(-\dfrac{\beta^{L-1}_{l_1...l_M}}{\sigma_a^2}\biggr),
		\end{align}
		where $Ei(.)$ denotes the exponential integral function.
		This completes the proof of Theorem~\ref{theorem_hop}.
	\end{proof}
	
	\subsection{Ad-hoc Relay Selection}
	Since, hop-by-hop relay selection does not involve rate optimization in the last hop, it cannot achieve full-diversity. As such, we next consider ad-hoc relay selection. Under this strategy, hop-by-hop relay selection is extended to achieve full-diversity by jointly combining the final two hops. Therefore, while the first $L-2$ relays are selected similar to the hop-by-hop relay selection, the last relay is selected such that the achievable rate of the last two hops is maximized.
	\begin{theorem}\label{theorem_adhoc}
		Consider a single-user, multi-hop relay network with $L$ hops and $M$ relay nodes per hop. For such a network, the achievable rate under ad-hoc relay selection is given by
		\begin{align*}%\label{eq_adhoc_rate}
			R^{\textrm{ad-hoc}} &=  -\dfrac{1}{\log 2}\sum_{i=1}^M \sum_{{\underset{= L-2}{l_1+...+{l_M}}}} \biggl[ A^{L-2}_{l_1...l_M}(i)\; e^{\beta^{L-2}_{l_1...l_M}(i)/\sigma_a^2} Ei\biggl(-\dfrac{\beta^{L-2}_{l_1...l_M}(i)}{\sigma_a^2}\biggr)\biggr],
		\end{align*}	
		where $\sigma_a^2\!=\!\frac{P\sigma_1^2}{\sigma^2}$, $A^{L-2}_{l_1...l_M}(i) \!=\! \binom{M}{i}\binom{L-2}{l_1,...,l_M} \biggl(\prod_{j=1}^{M}\binom{M}{j}^{l_j}\biggr)$ $(-1)^{i-1+\sum_{j=1}^{M}(j-1)l_j}$, $\beta^{L-2}_{l_1...l_M}(i) = 2i+\sum_{j=1}^{M}jl_j$ with $l_1,..,l_M$ denote the multinomial expansion exponents, $\binom{L-2}{l_1,...,l_M}$ and $\binom{M}{j}$ denote the multinomial and the binomial coefficients, respectively, and $Ei(.)$ is the exponential integral function.
	\end{theorem}
	\begin{proof}
		The CDF of $\gamma^{min}$ under ad-hoc relay selection can be expressed as, 
		\begin{align}\label{CDF_adhoc}
			F_{\gamma^{min}}^{\textrm{ad-hoc}}(x) =& 1-\biggl(1-\bigg(1- e^{(-2x/\sigma_a^2)}\bigg)^M\biggr) \bigg(1-\bigg(1-e^{(-x/\sigma_a^2)}\bigg)^M\bigg)^{L-2}.
		\end{align}
		Please refer to Appendix \ref{app7:4} for a detailed derivation of \eqref{CDF_adhoc}. Next, we use the multinomial expansion \cite{00102} and re-write \eqref{CDF_adhoc} as,
		\begin{align}
			F_{\gamma^{min}}^{\textrm{ad-hoc}}(x) =& 1\!-\!\sum_{i=1}^M \sum_{{\underset{= L-2}{l_1+...+{l_M}}}} \!\biggl[\!\! \binom{M}{i}\!\! \binom{L-2}{l_1,...,l_M} \!\!\biggl(\prod_{j=1}^{M}\!\!\binom{M}{j}^{l_j}\!\biggr)  (-1)^{i-1+\sum_{j=1}^{M}(j-1)l_j} e^{-(2i+\sum_{j=1}^{M}jl_j)x/\sigma_a^2}\biggr].
		\end{align}
		Taking the derivative with respect to $x$, we can derive the PDF of $\gamma^{min}$ as,
		\begin{align}
			f_{\gamma^{min}}^{\textrm{ad-hoc}}(x) 
			=& \sum_{i=1}^M \sum_{{\underset{= L-2}{l_1+...+{l_M}}}} \biggl[ A^{L-2}_{l_1...l_M}(i)\biggl(\dfrac{\beta^{L-2}_{l_1...l_M}(i)}{\sigma_a^2}\biggr) e^{-\beta^{L-2}_{l_1...l_M}(i)x/\sigma_a^2}\biggr].
		\end{align}
		As such the achievable rate can be derived using the variable transformation $t = 1+x$ and \cite[eq. (4.331.2)]{2014520} as,
		\begin{align}
			R^{\textrm{ad-hoc}} &= -\dfrac{1}{\log 2}\sum_{i=1}^M \sum_{{\underset{= L-2}{l_1+...+{l_M}}}} \biggl[ A^{L-2}_{l_1...l_M}(i)\; e^{\beta^{L-2}_{l_1...l_M}(i)/\sigma_a^2} Ei\biggl(-\dfrac{\beta^{L-2}_{l_1...l_M}(i)}{\sigma_a^2}\biggr)\biggr],
		\end{align}
		where $Ei(.)$ denotes the exponential integral function.
		This completes the proof of Theorem~\ref{theorem_adhoc}.
	\end{proof}
	
	\subsection{Block-by-Block Relay Selection}
	Next, we consider block-by-block relay selection to improve the performance further. Under this strategy, $L$ hops are divided into non-overlapping blocks of $w$ hops and the relays are selected such that the achievable rate of each block is maximized. An example of such a network, where the block size $w$ = 2 is illustrated in Fig. \ref{figure5}. With this strategy, it is important to ensure that the block size is selected such that the last block would be greater than one. Otherwise, there will be no rate optimization involved with the last hop where the destination node is fixed. We note that analyzing the achievable rate under block-by-block relay selection is difficult for $w>2$ due to the larger number of dependent links available within a given block that increases with $w$. However, $w=2$ is tractable and as such, we select the block size of two for analytical insights and derive the achievable rate in the following theorem.
	\begin{figure}
		\centerline{\includegraphics[width=0.75\textwidth]{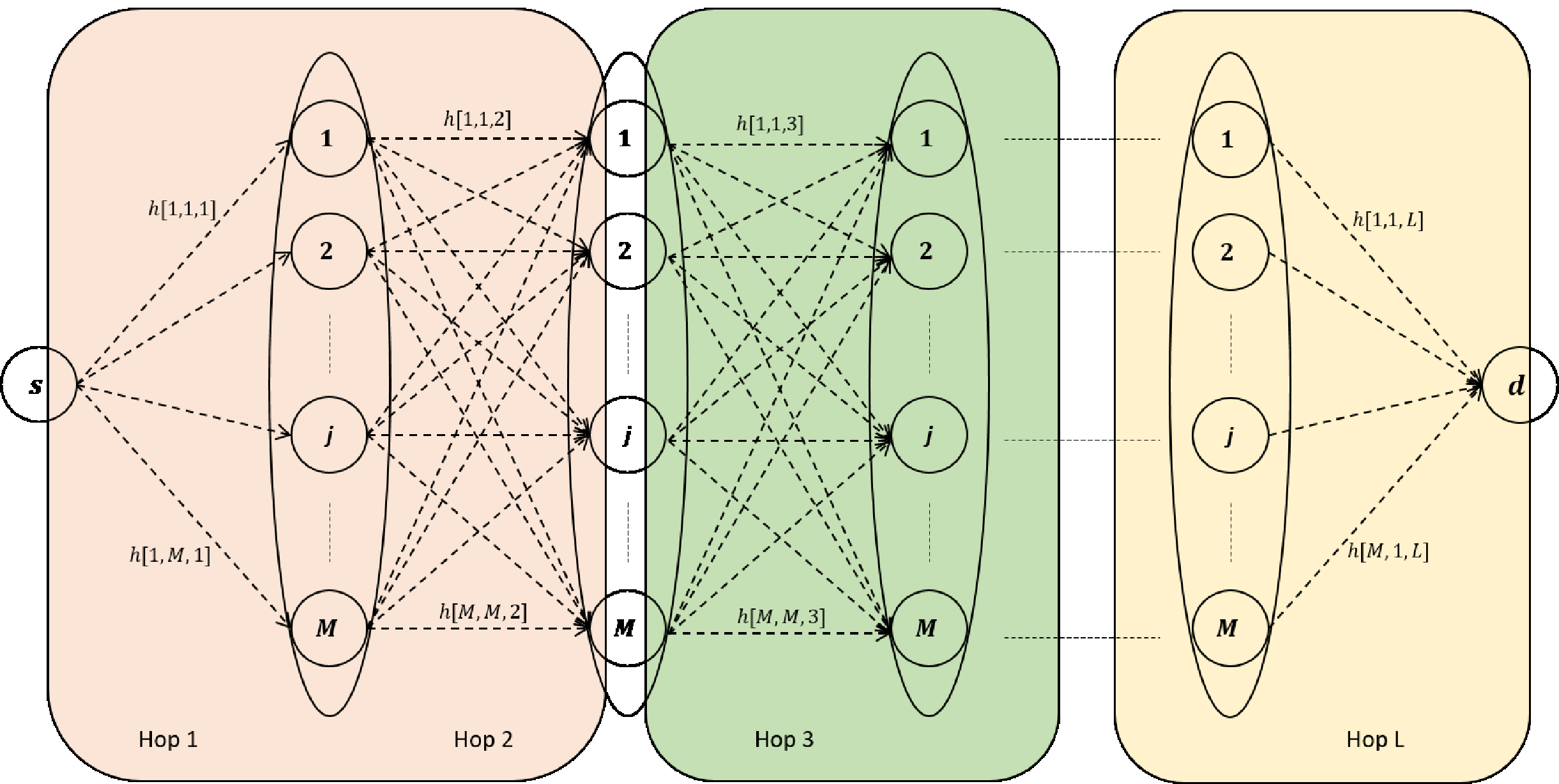}}
		\caption{Multi-hop relay network separated into non-overlapping blocks of $w=2$ hops}
		\label{figure5}
	\end{figure}
	\begin{theorem}\label{theorem_block}
		Consider a single-user, multi-hop relay network with $L$ hops and $M$ relay nodes per hop. For such a network, the achievable rate under block-by-block relay selection with $w=2$ can be expressed as,
		\begin{align*}%\label{eq_block_rate}
			R^{\textrm{block}} =& \dfrac{1}{\log 2}\sum_{i=1}^M \biggl\{ \sum^{T-1}_t \sum_{\underset{=Mt} {l_0+\dots+{l_M}}} \biggl[ A^{Mt}_{l_0\dots l_M}\!(i,t) e^{\beta^{Mt}_{l_0\dots l_M}(i)/\sigma_a^2} Ei\biggl(\dfrac{-\beta^{Mt}_{l_0\dots l_M}(i)}{\sigma_a^2}\biggr)\!\biggr] + \\ & \hspace{200pt} \binom{M}{i} (-1)^{i} e^{2i/\sigma_a^2} Ei\biggl(\!\dfrac{-2i}{\sigma_a^2}\biggr)\!\!\biggr\},
		\end{align*}	
		where $\sigma_a^2\!=\!\frac{P\sigma_1^2}{\sigma^2}, T\!=\!\left\lceil \frac{L}{w}\right\rceil$,  $A^{Mt}_{l_0...l_M}(i,t) \!=\! \binom{M}{i} \binom{T-1}{t}$ $ \binom{Mt}{l_0,...,l_M}\biggl(\prod_{j=1}^{M}\binom{M}{j}^{l_j}\biggr)(-1)^{i+t+\sum_{j=1}^{M}jl_j}$, $\beta^{Mt}_{l_0...l_M}(i) = 2i+\sum_{j=1}^{M}(j+1)l_j$ with $l_0,..,l_M$ denote the multinomial expansion exponents, $\binom{Mt}{l_0,...,l_M}$ and $\binom{M}{j}$ denote the multinomial and the binomial coefficients, respectively, and $Ei(.)$ is the exponential integral function. .
	\end{theorem}
	\begin{proof}
		The CDF of $\gamma^{min}$ under block-by-block relay selection with $w=2$ can be expressed as \begin{align}
			F_{\gamma^{min}}^{\textrm{block}}(x) &= 1\!-\!\biggl(1\!-\!\bigg(1\!- e^{(-2x/\sigma_a^2)}\bigg)^M\biggr)\bigg(1\!-\!\bigg(1-e^{(-x/\sigma_a^2)}\bigg(1-(1-e^{(-x/\sigma_a^2)})^M\bigg)\bigg)^M\bigg)^{T\!-\!1}. \label{CDF_block} 
		\end{align}
		Please refer to Appendix \ref{app7:5} for a detailed derivation of \eqref{CDF_block}. 	
		Next, we use the multinomial expansion \cite{00102} and re-write the CDF as,
		\begin{align}
			F_{\gamma^{min}}^{\textrm{block}}(x) &= \sum_{i=1}^M\! \binom{M}{i}\!(-1)^i e^{-2ix/\sigma_a^2} \!+\! \sum_{i=1}^M \sum_{t=1}^{T-1} \sum_{\underset{=Mt} {l_0+...+{l_M}}}\!\biggl[\!\binom{M}{i}\binom{T-1}{t} \binom{Mt}{l_0,...,l_M}\biggl(\prod_{j=1}^{M}\binom{M}{j}^{l_j}\biggr)\nonumber \\ & \hspace{175pt}(-1)^{i+t+\sum_{j=1}^{M}jl_j}\;  e^{-(2i+\sum_{j=1}^{M}(j+1)l_j)x/\sigma_a^2}\biggr]+1. \label{eq_16}
		\end{align}
		Taking the derivative with respect to $x$, we can derive the PDF of $\gamma^{min}$ as, 
		\begin{align}
			f_{\gamma^{min}}^{\textrm{block}}(\!x\!) 
			&= -\!\sum_{i=1}^M \!\!\binom{M}{i} (-1)^i \bigg(\!\dfrac{2i}{\sigma_a^2}\!\bigg) e^{-2ix/\sigma_a^2} \!-\!\sum_{i=1}^M \sum_{t=1}^{T-1} \! \sum_{\underset{=Mt} {l_0+...+{l_M}}} A^{Mt}_{l_0...l_M}(i,t)  \bigg(\!\dfrac{\beta^{Mt}_{l_0...l_M}(i)}{\sigma_a^2}\!\bigg) e^{-\beta^{Mt}_{l_0...l_M}(i)x/\sigma_a^2}. \label{eq_17} 
		\end{align}
		As such the achievable rate can be derived using the variable transformation $t = 1+x$ and \cite[eq. (4.331.2)]{2014520} as, 
		\begin{align}
			R^{\textrm{block}} &= \dfrac{1}{\log 2}\sum_{i=1}^M \biggl\{\sum^{T-1}_t \sum_{\underset{=Mt} {l_0+\dots+{l_M}}} A^{Mt}_{l_0\dots l_M}(i,t) e^{\beta^{Mt}_{l_0\dots l_M}(i)/\sigma_a^2} Ei\biggl(\dfrac{-\beta^{Mt}_{l_0\dots l_M}(i)}{\sigma_a^2}\biggr)+\nonumber\\ & \hspace{200pt}\binom{M}{i} (-1)^{i} e^{2i/\sigma_a^2} Ei\biggl(\dfrac{-2i}{\sigma_a^2}\biggr)\biggr\}. \label{eq_18}
		\end{align}
		where $Ei(.)$ denotes the exponential integral function. This completes the proof of Theorem~\ref{theorem_block}. 
	\end{proof} 	
	
	\subsection{Sliding Window based Relay Selection}
	Finally, we consider sliding window based relay selection to remove the dependency of block size on the number of hops and to improve the performance further. Under this strategy, we consider a sliding window of $w$ hops to determine the relay selection in the first hop of the window. For example, we start by considering the first $w$ hops and find the relay selection such that the achievable rate in those $w$ hops is maximized. However, we only fix the relay selected in the first hop. Next, we consider $w$ hops from the second hop to $w+1$ and fix the relay selected for the second hop. We continue this until relay selection is fixed for first $L-w$ hops. Then we consider the last $w$ hops and fix the relays selected for all of them. Similar to block-by-block relay selection, we consider the window size of two to obtain analytical insights under sliding window based relay selection. However, it is important to note that, under this relay selection strategy, one window of $w$ hops is not independent of the next window as they consist of at least one common hop, as illustrated in Fig. \ref{figure6}. Therefore, it is difficult to obtain an exact expression for the achievable rate. Thus, we take a more tractable approach by considering an approximate relay network with independent windows.
	\begin{figure}
		\centerline{\includegraphics[width=0.75\textwidth]{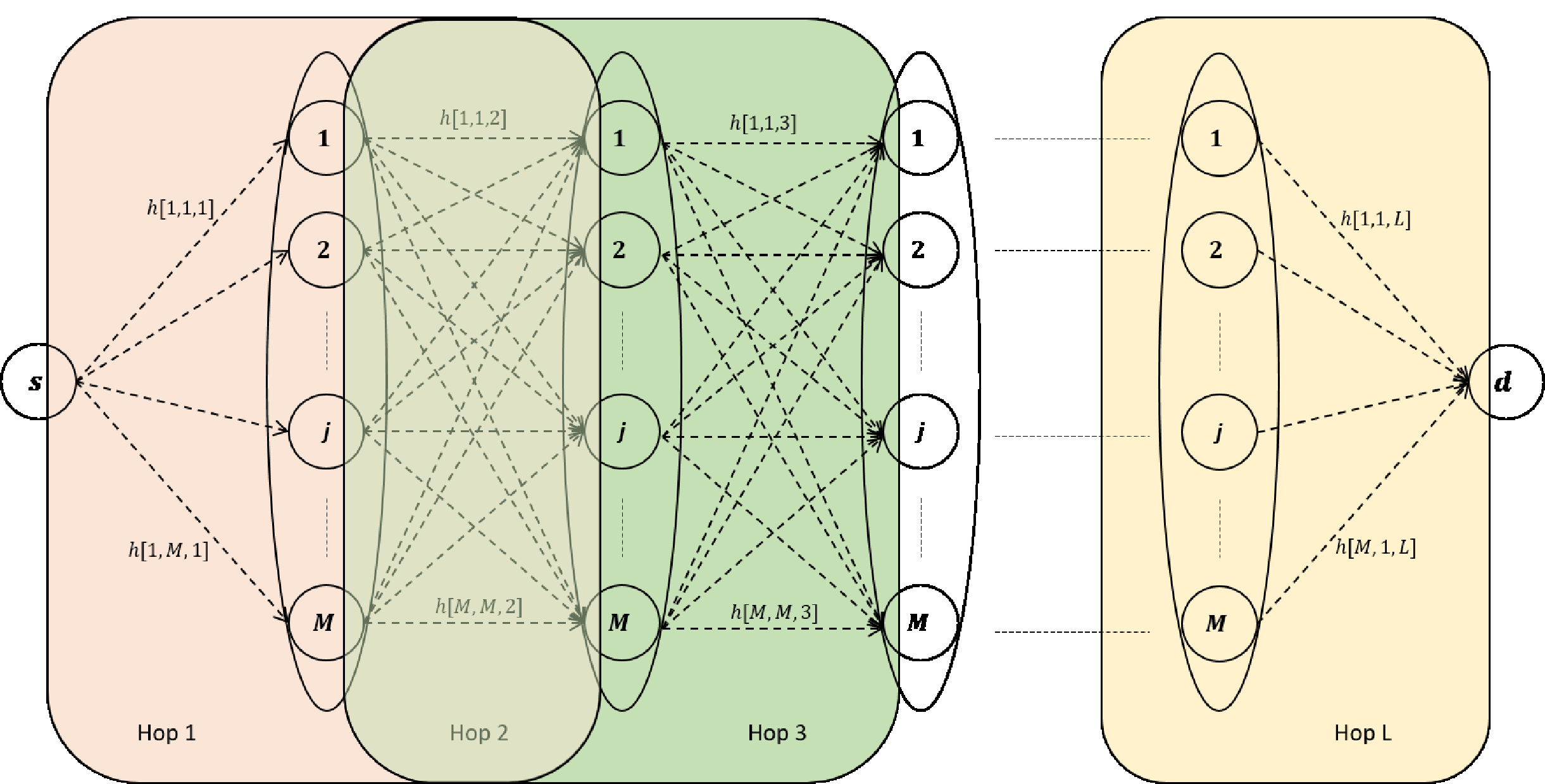}}
		\caption{Multi-hop relay network with sliding window of $w=2$ hops}
		\label{figure6}
	\end{figure}
	
	\subsubsection*{Independent Window Model}
	In this section we consider an approximated model for sliding window based relay selection such that each window is independent to every other window. Under this model, there are $L-w+1$ independent windows. As such, we can derive the achievable rate under the independent window assumption as follows.
	\begin{lemma}\label{theorem_sliding}
		Consider a single-user, multi-hop relay network with $L$ hops and $M$ relay nodes per hop. Under the independent window assumption described above, the achievable rate with sliding window based relay selection and $w=2$ can be expressed as,
		\begin{align*}%\label{eq_sliding_rate}
			R^{\textrm{sliding}}_{\textrm{ind}} &= \dfrac{1}{\log 2}\sum_{i=1}^M \sum_{\underset{=T-1} {l_1+\dots+{l_{\bar{Q}}}}} A^{T-1}_{l_1...l_{\bar{Q}}}(i) \; e^{\beta^{T-1}_{l_1...l_{\bar{Q}}}(i)/\sigma_a^2}Ei\biggl(-\dfrac{\beta^{T-1}_{l_1...l_{\bar{Q}}}(i)}{\sigma_a^2}\biggr),
		\end{align*}	
		where $\sigma_a^2=\frac{P\sigma_1^2}{\sigma^2}, T=L-w+1$, $\bar{Q}=M^2(M^2\!-\!1)$, $A^{T-1}_{l_1...l_{\bar{Q}}}(i) \!=\! \binom{M}{i} \binom{T-1}{l_1,...,l_{\bar{Q}}}\! \biggl(\prod_{q=1}^{\bar{Q}}(a_q)^{l_q}\biggr)(-1)^{i}$, $\beta^{T-1}_{l_1...l_{\bar{Q}}}(i) = 2i+\sum_{q=1}^{\bar{Q}}b_ql_q$ with $l_1,..,l_{\bar{Q}}$ denote the multinomial expansion exponents, $\binom{T-1}{l_1,...,l_{\bar{Q}}}$ and $\binom{M}{j}$ denote the multinomial and the binomial coefficients, respectively, $Ei(.)$ is the exponential integral function and the coefficients $a_q$ and $b_q$ can be numerically computed for a given $M$ by solving the integral
		\begin{align*}
			\sum_{q=1}^{\bar{Q}} a_q\; e^{-b_q\;x/\sigma_a^2} &= 1\!-\!(1\!-\!e^{-x/\sigma_a^2})^M - \!\sum_{N=1}^{M-1} P^M_N e^{-(M-N)x/\sigma_a^2} \bigg[\int_{\gamma_{11}^t =0}^{x}\!\!\!\!\!\!e^{-\gamma_{11}^t}\int_{\gamma_{12}^t =\gamma_{11}^t}^{x}\!\!\!\!\!\!\!\!\!e^{-\gamma_{12}^t}\dots\int_{\gamma_{1N}^t =\gamma_{1,N-1}^t}^{x}\!\!\!\!\!\!\!\!\!\!\!\!e^{-\gamma_{1N}^t}\nonumber \\ &\quad  \sum_{k=1}^{N}(M-N)\bigg(\!\dfrac{(1\!-\!e^{-\gamma_{1k}^t/\sigma_a^2})^{M(M-k)}}{M\!-\!k}- \dfrac{(1\!-\!e^{-\gamma_{1k}^t/\sigma_a^2})^{M(M-k+1)}}{M\!-\!k\!+\!1} \!\bigg)  d\gamma_{11}^t \dots d\gamma_{1N}^t \bigg],
		\end{align*}
		with $P^M_N$ denoting the number of possible permutations when $N$ units are selected from $M$.
	\end{lemma}
	\begin{proof}
		The CDF of $\gamma^{min}$ under sliding window based relay selection with $w=2$ and independent window assumption can be expressed as,  
		\begin{align} 
			F_{\gamma^{min}}^{\textrm{sliding}}(x) 
			&= 1-\biggl(1-\bigg(1- e^{(-2x/\sigma_a^2)}\bigg)^M\biggr)\biggl(1-(1-e^{-x/\sigma_a^2})^M - \sum_{N=1}^{M-1} P^M_N e^{-(M-N)x/\sigma_a^2} \nonumber \\ & \qquad \bigg(\int_{\gamma_{11}^t =0}^{x}\int_{\gamma_{12}^t =\gamma_{11}^t}^{x}...\int_{\gamma_{1N}^t =\gamma_{1,N-1}^t}^{x}\sum_{k=1}^{N}P(\gamma_{max}^{(t)}=\gamma_{1k}^t)\; e^{-\gamma_{11}^t}... e^{-\gamma_{1N}^t} d\gamma_{11}^t ... d\gamma_{1N}^t \bigg)\!\biggr)^{T\!-\!1}, \label{CDF_sliding} 
		\end{align}
		where $\sum_{k=1}^{N}P(\gamma_{max}^{(t)}=\gamma_{1k}^t)$ for the region $\gamma_{11}^t < \gamma_{12}^t < ... < \gamma_{1N}^t \le x$ is given by 
		\begin{align}
			\sum_{k=1}^{N}P(\gamma_{max}^{(t)} = \gamma_{1k}^t) &= \sum_{k=1}^{N} \bigg(\!\dfrac{M\!-\!N}{M\!-\!k}\!\bigg) (1-e^{-\gamma_{1k}^t/\sigma_a^2})^{M(M-k)} \!-\! \bigg(\!\dfrac{M\!-\!N}{M\!-\!k\!+\!1}\!\bigg) (1-e^{-\gamma_{1k}^t/\sigma_a^2})^{M(M-k+1)}. \label{eq_21}
		\end{align}
		Please refer to Appendix \ref{app7:6} for a detailed derivation of \eqref{CDF_sliding}. Due to the complexity of the above integration, it is not possible to obtain an analytical expression for the CDF of $\gamma^{min}$ in general. However, for a given $M$ this can be solved and the CDF will take the form of 
		\begin{align}
			F_{\gamma^{min}}^{\textrm{sliding}}(x) 
			&\!=\!1\!-\!\biggl(\!1\!-\!\bigg(\!1\!-\! e^{(-2x/\sigma_a^2)}\!\bigg)^{\!M}\biggr)\!\biggl(\!\sum_{q=1}^{\bar{Q}}\!a_q e^{-b_qx/\sigma_a^2}\! \biggr)^{\!T\!-\!1}\!\!,
		\end{align}
		where the coefficients $a_q$ and $b_q$ can be numerically computed for a given $M$ by solving the integration given in \eqref{CDF_sliding}. % with $\bar{Q}=M^2(M^2-1)$. 
		Next, we use the multinomial expansion \cite{00102} and take the derivative with respect to $x$ to derive the PDF of $\gamma^{min}$ as,
		\begin{align}\label{eq_sliding_pdf}
			f_{\gamma^{min}}^{\textrm{sliding}}(x) &= -\sum_{i=1}^M \sum_{\underset{=T-1} {l_1+\dots+{l_{\bar{Q}}}}} A^{T-1}_{l_1...l_{\bar{Q}}}(i) \biggl(\dfrac{\beta^{T-1}_{l_1...l_{\bar{Q}}}(i)}{\sigma_a^2}\biggr) e^{-\beta^{T-1}_{l_1...l_{\bar{Q}}}(i)x/\sigma_a^2}.
		\end{align}
		As $\log_2(1+x)$ is an increasing function of $x$, the achievable rate can be derived using the variable transformation $t = 1+x$ and \cite[eq. (4.331.2)]{2014520} as, 
		\begin{align}
			R^{\textrm{sliding}}_{\textrm{ind}} &= \dfrac{1}{\log 2}\sum_{i=1}^M \sum_{\underset{=T-1} {l_1+\dots+{l_{\bar{Q}}}}} A^{T-1}_{l_1...l_{\bar{Q}}}(i)\; e^{\beta^{T-1}_{l_1...l_{\bar{Q}}}(i)/\sigma_a^2}Ei\biggl(-\dfrac{\beta^{T-1}_{l_1...l_{\bar{Q}}}(i)}{\sigma_a^2}\biggr), \label{eq_24}
		\end{align}
		where $Ei(.)$ denotes the exponential integral function. This completes the proof of Lemma \ref{theorem_sliding}.
	\end{proof}	
	As such, we can write an approximation on the achievable rate under the sliding window based relay selection using Lemma \ref{theorem_sliding} as,
	\begin{align}\label{eq_sliding_rate}
		R^{\textrm{sliding}} &\approx R^{\textrm{sliding}}_{\textrm{ind}}. 
	\end{align}
	In the following example, we analyze the accuracy of the above approximation.
	
	\vspace{0.25cm}
	\noindent
	\textbf{Example 3}: Consider a single-user, multi-hop relay network where $M=2, L=3,5$. Assuming constant distance in all hops, we normalize the variance of the channel fading coefficient, $\sigma_1^2$, to unity \cite{2969232}. For such a network, the analytical and the simulated achievable rate against the received SNR are plotted in Fig. \ref{figure4} under the sliding window based relay selection strategy with $w=2$. From the figure, it can be observed that the analytical rate expression achieved with \eqref{eq_sliding_rate} is accurate when $L$ is small. While, the approximation becomes lose as we increase $L$, most of the existing relay networks consist of few hops \cite{0286-3,1023886}. As such, \eqref{eq_sliding_rate} can be used to accurately approximate the achievable rate under the sliding window based relay selection.
	\begin{figure}
		\centerline{\includegraphics[width=0.7\textwidth]{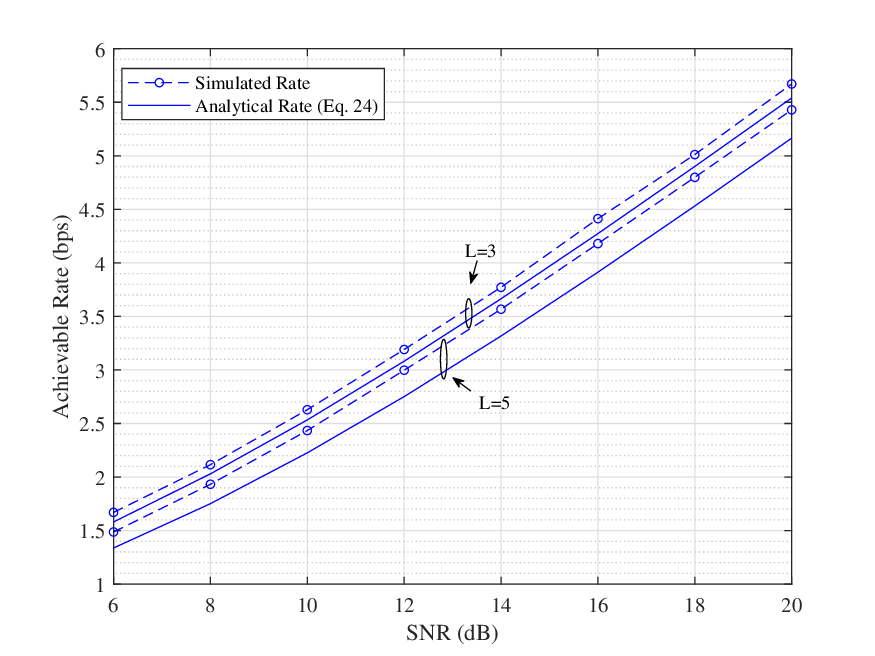}}
		\caption{The achievable rate against received SNR under the sliding window based relay selection when $M=2, w=2$}
		\label{figure4}
	\end{figure}
	
	We note that increasing $w$ improves the performance of sliding window based relay selection. However, for a given $M$ the complexity of sliding window based relay selection exponentially increases proportional to $w$. As such, in the following example, we proceed to analyze the impact of $w$ on the performance of sliding window based relay selection such that we can achieve the best performance with a considerable amount of complexity.
	
	\begin{table}
		\centering
		\caption{Effectiveness of $w$ for a given $M$ when $L=6$}
		\label{table3}
		\begin{tabular}{|p{0.5cm}||p{0.75cm}|p{0.75cm}|p{0.75cm}|p{0.75cm}|p{0.75cm}|p{0.75cm}|p{0.75cm}|p{0.75cm}|p{0.75cm}|p{0.75cm}|p{0.75cm}|p{0.75cm}|}\hline
			$M$ & $w$=1 & $w$=2 & $w$=3 & $w$=4 & $w$=5 & $w$=6 \\ \hline  				
			2 & 80.50 & 96.24 & 99.17 & 99.80 & 99.97 & 100 \\ \hline
			3 & 76.83 & 94.45 & 98.34 & 99.59 & 99.94 & 100 \\ \hline 
			4 & 73.89 & 92.95 & 97.82 & 99.30 & 99.78 & 100 \\ \hline 
			5 & 72.98 & 92.87 & 97.79 & 99.25 & 99.81 & 100 \\ \hline 
			6 & 71.58 & 92.44 & 97.78 & 99.20 &	99.80 & 100 \\ \hline 
			7 & 68.16 & 92.00 & 97.29 & 99.04 &	99.75 & 100 \\ \hline 
		\end{tabular}
	\end{table}
	\vspace{0.25cm}
	\noindent
	\textbf{Example 4}: Consider a single-user, multi-hop relay network where the channels between nodes follow a Rayleigh distribution with zero mean and unit variance. For such a network, the percentage of the achievable rate obtained based on sliding window based relay selection compared to that of optimal relay selection, which we define as the effectiveness of $w$, is given in Table \ref{table3} and Table \ref{table2}. Results are obtained by averaging over $5000$ simulations. 
	From Table \ref{table3}, we observe that as $M$ increases, the percentage of the achievable rate compared to the optimal rate decreases for a given $w$. However, even for a multi-relay network of $M=7$ relays in each hop, $w=3$ manages to result in more than 97\% of the optimal achievable rate. Similarly, from Table \ref{table2}, we observe that as $L$ increases, it is better to increase the window size $w$ as the effectiveness of a given $w$ decreases. However, even for a multi-hop network of $L=12$ hops, $w=3$ manages to result in more than 98\% of the optimal achievable rate.
	As such, $w=2$ provides closer to 90\% of accuracy where as $w=3$ provides closer to 97\% of accuracy compared to the optimal solution. Further, the drop of accuracy with increasing $L$ and $M$ is considerably lower for $w=3$ compared to $w=2$. Therefore, we can conclude that the window size $w=3$ provides an acceptable level of accuracy compared to the optimal solution with a considerable level of complexity.
	\begin{table*}
		%\centering\small
		\caption{Effectiveness of $w$ for a given $L$ when $M=2$}
		\begin{tabular}{|p{0.5cm}||p{0.75cm}|p{0.75cm}|p{0.75cm}|p{0.75cm}|p{0.75cm}|p{0.75cm}|p{0.75cm}|p{0.75cm}|p{0.75cm}|p{0.75cm}|p{0.75cm}|p{0.75cm}|}\hline
			$L$ & $w$=1 & $w$=2 & $w$=3 & $w$=4 & $w$=5 & $w$=6 & $w$=7 & $w$=8 & $w$=9 & $w$=10 & $w$=11 & $w$=12\\ \hline  				
			3 & 85.81 & 98.20 & 100 & & & & & & & & & \\ \hline
			4 & 85.08 & 96.93 & 99.49 & 100 & & & & & & & &	\\ \hline
			5 & 83.52 & 96.75 & 99.64 & 99.93 & 100	& & & & & & & \\ \hline 
			6 & 81.64 & 95.84 & 99.26 & 99.90 & 99.98 & 100 & & & & & & \\ \hline 
			7 & 79.05 & 94.97 & 98.84 & 99.86 & 99.92 & 99.98 & 100 & & & & & \\ \hline 
			8 & 79.33 & 94.65 & 98.75 & 99.65 & 99.84 & 99.99
			& 100 & 100	& & & &	\\ \hline 
			9 & 76.92 & 94.82 & 99.11 & 99.69 & 99.98 & 99.99 & 100 & 100 & 100 & & & \\ \hline 
			10 & 76.06 & 94.07 & 98.75 & 99.67 & 99.90 & 99.94 & 100 & 100 & 100 & 100 & & \\ \hline
			11 & 75.51 & 93.65 & 98.48 & 99.74 & 99.98 & 99.99
			& 100 & 100 & 100 & 100 & 100 & \\ \hline
			12 & 73.79 & 93.52 & 98.43 & 99.67 & 99.97	& 99.99 & 100 & 100 & 100 & 100 & 100 & 100	\\ \hline 
		\end{tabular}
		\label{table2}
	\end{table*}
	
	\noindent In the following example, we further analyze the achievable rate obtained by the five relay selection strategies and their derived analytical expressions.
	
	\vspace{0.25cm}
	\noindent
	\textbf{Example 5}: Consider a single-user, multi-hop relay network where $M=2, L=4$. Assuming constant distance in all hops, we normalize the variance of the channel fading coefficient, $\sigma_1^2$, to unity \cite{2969232}. For such a network, the analytical and the simulated achievable rate against the received SNR are plotted in Fig. \ref{figure9} under optimal relay selection, hop-by-hop relay selection, ad-hoc relay selection, block-by-block relay selection with block size of two and sliding window based relay selection with window size of two. From the figure, we observe that sliding window based relay selection results in the highest achievable rate compared to other sub-optimal relay selection strategies. We also, observe that the approximate analytical rate in \eqref{optimal_approx2} is higher than the simulated optimal rate while approximate analytical rate in \eqref{eq_sliding_rate} is less than the simulated rate under sliding window based relay selection. As expected, the analytical rate expressions achieved with Theorem \ref{theorem_hop}, Theorem \ref{theorem_adhoc} and Theorem \ref{theorem_block} are similar to the simulated rates under hop-by-hop, ad-hoc and block-by-bloc relay selection strategies, respectively. We also note that the computation time of analytical rates are independent of the received SNR. As such, the average computational time of analytical rates are $0.2492, 0.3896, 0.2302, 0.0944$ and $0.0616$ for optimal relay selection in \eqref{optimal_approx2}, sliding window based relay selection in \eqref{eq_sliding_rate}, block-by-block relay selection in Theorem \ref{theorem_block}, ad-hoc relay selection in Theorem \ref{theorem_adhoc} and hop-by-hop relay selection in Theorem \ref{theorem_hop}, respectively. As such, we can observe that the analytical approximation for sliding window based relay selection has the highest complexity while analytical rate for hop-by-hop relay selection has the lowest complexity when $L=4$. However, we note that for larger $L$ values, the complexity of analytical approximation in \eqref{optimal_approx2} would increase significantly due to consideration of all hops compared to the increment in computation time for other sub-optimal relay selection strategies.
	\begin{figure}
		\centerline{\includegraphics[width=0.7\textwidth]{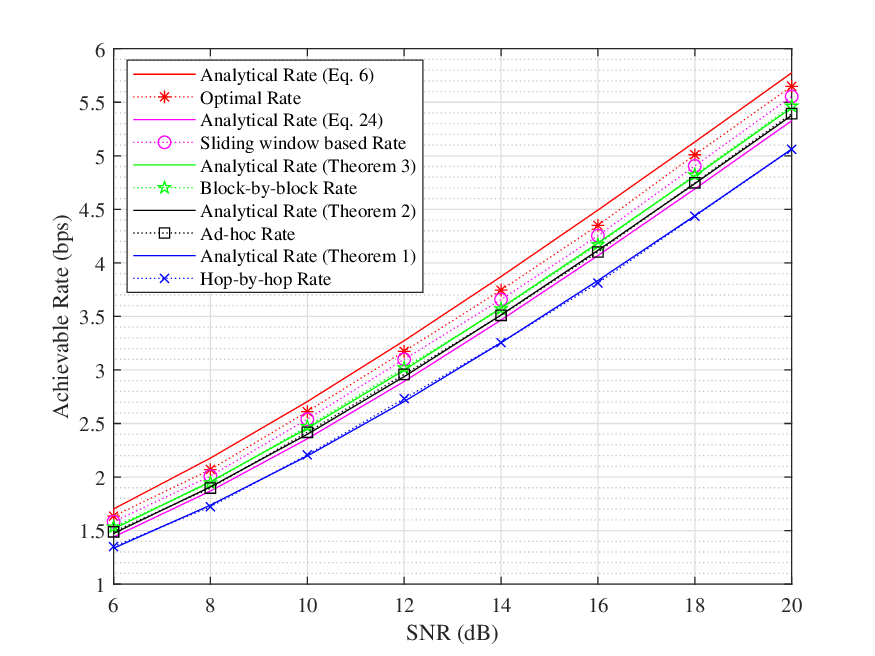}}
		\caption{The achievable rate against received SNR under different relay selection strategies}
		\label{figure9}
	\end{figure}
	
	\section{Multi-User Relay Network} \label{section-multi}
	In this section we extend the achievable rate analysis of the five relay selection strategies adopted in this paper to a multi-user multi-hop relay network and derive approximate expressions for the achievable sum-rate. Due to use of orthogonal transmission among different users and sufficient interference cancellation, we can reasonably assume that the interference caused by other transmitting nodes is negligible compared to the noise power at any receiver node \cite{2809748,2007.647}. As such, we consider a noise limited multi-user network where the relay selection only depends on the received SNR. Furthermore, when the number of available relays is significantly large compared to the number of users, the probability of two users selecting the same relay node in a given hop is small. Therefore, we can consider that the impact of other users on the relay selection of a given user is negligible. As such, for a noise limited multi-hop relay network with $N$ users, the achievable sum-rate can be approximated as, 
	\begin{align} 
		&R^{\textrm{opt}} \approx \dfrac{N}{\log 2}\sum_{q=1}^{Q} a_q \; e^{b_q/\sigma_a^2}\; Ei\biggl(-\dfrac{b_q}{\sigma_a^2}\biggr),  \label{eq_multi_optimal_rate} \\ 
		&R^{\textrm{hop}} \approx -\dfrac{N}{\log 2}\sum_{l_1+...+l_M = L-1} A^{L-1}_{l_1...l_M} e^{\beta^{L-1}_{l_1...l_M}/\sigma_a^2} Ei\biggl(-\dfrac{\beta^{L-1}_{l_1...l_M}}{\sigma_a^2}\biggr), \label{eq_multi_hop_rate} \\ 
		&R^{\textrm{ad-hoc}} \approx -\dfrac{N}{\log 2}\sum_{i=1}^M \sum_{l_1+...+l_M = L-2} A^{L-2}_{l_1...l_M}(i) e^{\beta^{L-2}_{l_1...l_M}(i)/\sigma_a^2} Ei\biggl(-\dfrac{\beta^{L-2}_{l_1...l_M}(i)}{\sigma_a^2}\biggr), \label{eq_multi_ad-hoc_rate} \\ 
		&R^{\textrm{block}} \approx \dfrac{N}{\log 2}\!\sum_{i=1}^M \biggl\{\sum^{T-1}_t \sum_{\underset{=Mt} {l_0+\dots+{l_M}}} A^{Mt}_{l_0\dots l_M}(i,t) e^{\beta^{Mt}_{l_0\dots l_M}(i)/\sigma_a^2} Ei\biggl(\dfrac{-\beta^{Mt}_{l_0\dots l_M}(i)}{\sigma_a^2}\biggr)+ \nonumber \\ & \hspace{275pt}\binom{M}{i} (-1)^{i} e^{2i/\sigma_a^2} Ei\biggl(\dfrac{-2i}{\sigma_a^2}\biggr)
		\biggr\}, \label{eq_multi_block_rate} \\
		&R^{\textrm{sliding}} \approx \dfrac{N}{\log 2}\sum_{i=1}^M \sum_{l_1+...+l_Q = T-1} A^{T-1}_{l_1...l_Q}(i) e^{\beta^{T-1}_{l_1...l_Q}(i)/\sigma_a^2} Ei\biggl(-\dfrac{\beta^{T-1}_{l_1...l_Q}(i)}{\sigma_a^2}\biggr).\label{eq_multi_sliding_rate}
	\end{align}
	by multiplying the single user achievable rate with the number of users. In the following example, we have further illustrated the accuracy of the above approximations.
	
	\vspace{0.25cm}
	\noindent
	\textbf{Example 6}: 
	\begin{figure}
		\centerline{\includegraphics[width=0.7\textwidth]{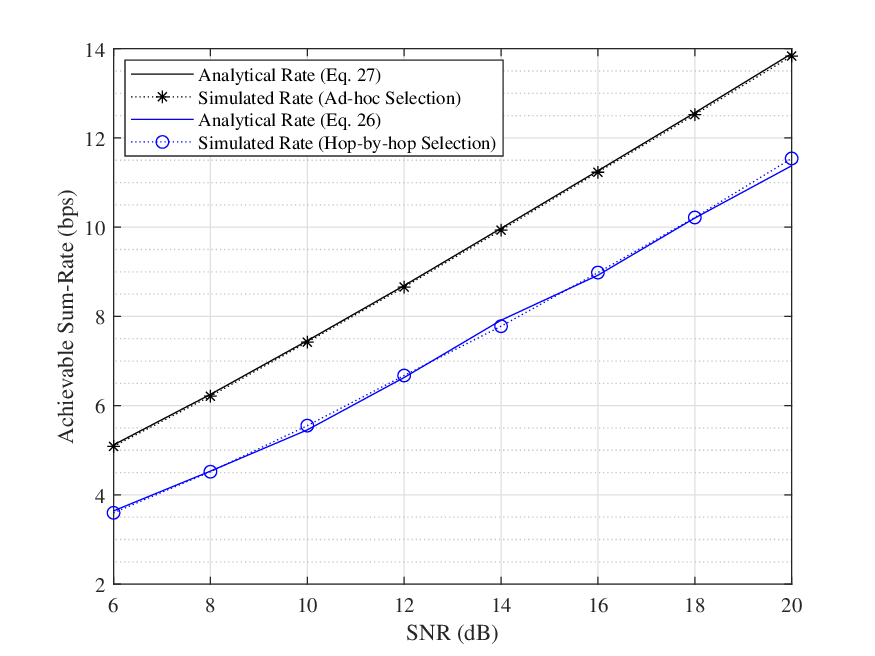}}
		\caption{The achievable sum-rate against the received SNR when $N=2, M=10, L=6$}
		\label{figure753}
	\end{figure}
	Consider a noise limited multi-hop relay network where $N=2, M=10, L=6$. Assuming constant distance in all hops, we normalize the variance of the channel fading coefficient, $\sigma_1^2$, to unity \cite{2969232}. For such a network, the analytical and the simulated achievable sum-rate against the received SNR are plotted in Fig. \ref{figure753} under hop-by-hop and ad-hoc relay selection strategies.
	From the figure, it can be observed that the analytical sum-rates given in \eqref{eq_multi_hop_rate} and \eqref{eq_multi_ad-hoc_rate} are similar to the simulated achievable sum-rate under hop-by-hop and ad-hoc relay selection strategies, respectively. Even though, not shown here due to computational complexity of the multinomial expansion exponents, similar observation can also be expected for block-by-block relay selection strategy. 	
	
	\section{Conclusion} \label{section-conclusion}
	The achievable rate optimization problem is analyzed for a single-user, multi-hop relay network consisting of multiple relay nodes. We consider five relay selection strategies including optimal relay selection which in general has an exponential complexity in respect to the number of hops. We analyzed optimal relay selection using the dynamic programming based Viterbi algorithm and derived an approximation on the optimal achievable rate which is more accurate than the simple approximation based on the independent path assumption. Next, we considered four sub-optimal relay selection strategies, namely hop-by-hop, ad-hoc, block-by-block and sliding window based relay selection. We obtained exact analytical expressions for the achievable rate under the first three relay selection strategies and extended sliding window based relay selection to DF relay network to derive an approximation on the achievable rate with the window size of two. Further, we showed that sliding window based relay selection with a window size of three is sufficient to achieve more than 97\% of the optimal achievable rate. Finally, we extended this analysis to a noise limited multi-user network and showed that when there is a larger number of relay nodes per hop compared to the the number of users, the achievable sum-rate can be approximated by the single user rate times the number of users.  
	
	In the analytical results presented in this paper we assumed that each link in the network follows a Rayleigh fading distribution and same large scale fading. A desirable future extension would be to consider different fading channel distributions such as Nakagami-m, alpha-mu and channels with non-identical large scale fading resulted from various source-relay-destination distances. It would also be interesting to generalize these results to an interference limited multi-user network.
	
	\appendices
	\section{Proof of Lemma 1}\label{app7:2}
	Let us define $\gamma_{ij}=\frac{P}{\sigma^2} |h[j,i,l]|^2$ as the received SNR of the relay link between the transmitting node $i$ and the receiving node $j$. Since, each link in the network follows a Rayleigh distribution with zero mean and variance $\sigma_1^2$, we write the cumulative distribution function (CDF) of $\gamma_{ij}$ as,
	\begin{align}
		F_{\gamma_{ij}}(x) = 1-e^{(-x/\sigma_a^2)}, \; \forall i,j,
	\end{align}
	where $\sigma_a^2=\frac{P\sigma_1^2}{\sigma^2}$. Let $Y_q$ be the end-to-end received SNR of path $q$, which is taken as the minimum SNR over all the hops in path $q$. Therefore, the CDF of $Y_q$ can be written as,
	\begin{align}
		F_{Y_q}(x) = 1-e^{(-L\,x/\sigma_a^2)}.
	\end{align}
	Therefore, $\gamma^{min}$ can be computed as the maximum end-to-end received SNR over all the paths. Under this approximated model, all $Q=M^{L-1}$ paths are independent. As such, we can derive the CDF of $\gamma^{min}$ as,
	\begin{align}
		F_{\gamma^{min}}^{\textrm{opt}}(x) &= \bigg(1-e^{(-L\,x/\sigma_a^2)}\bigg)^Q.
	\end{align}
	Next, we use the binomial expansion and re-write
	\begin{align}
		F_{\gamma^{min}}^{\textrm{opt}}(x) = \sum_{q=0}^Q \binom{Q}{q}(-1)^{q} e^{-L\,q\,x/\sigma_a^2},
	\end{align}
	where $\binom{Q}{q}$ denotes the binomial coefficient. Taking the derivative with respect to $x$, we can derive the PDF of $\gamma^{min}$ as,
	\begin{align}\label{Optimal_PDF}
		f_{\gamma^{min}}^{\textrm{opt}}(x) = \sum_{q=1}^Q \binom{Q}{q}(-1)^{q}\biggl(\dfrac{-Lq}{\sigma_a^2}\biggr) e^{-L\,q\,x/\sigma_a^2}.
	\end{align}
	As $\log_2(1+x)$ is an increasing function of $x$, the achievable rate can be derived using the variable transformation $t = 1+x$ and \cite[eq. (4.331.2)]{2014520} as,
	\begin{align}
		R^{\textrm{opt}}_{\textrm{ind}} &= \int_{0}^{\infty} \log_2(1+x)f_{\gamma^{min}}^{\textrm{opt}}(x) dx \nonumber \\ 
		&{\overset{m}{=}} \! \sum_{q=1}^Q \!\!\binom{Q}{q}\! (-1)^{q}\biggl(\!\dfrac{-Lq}{\sigma_a^2}\!\biggr)\!\dfrac{1}{\log 2} e^{Lq/\sigma_a^2} \!\! \int_{1}^{\infty} \!\!\! \log(t)e^{-Lqt/\sigma_a^2} dt \nonumber \\ 
		&{\overset{n}{=}} \sum_{q=1}^Q \!\!\binom{Q}{q}\! (-1)^{q}\biggl(\!\dfrac{-Lq}{\sigma_a^2}\!\biggr)\!\dfrac{1}{\log 2}\, e^{L\,q/\sigma_a^2} \biggl[-\dfrac{\sigma_a^2}{Lq}Ei\biggl(\!-\dfrac{Lq}{\sigma_a^2}\biggr)\biggr] \nonumber \\ 
		&= \dfrac{1}{\log 2}\sum_{q=1}^Q \binom{Q}{q} (-1)^{q}\,e^{L\,q/\sigma_a^2}\; Ei\biggl(-\dfrac{Lq}{\sigma_a^2}\biggr),
	\end{align}
	where in step $m$, we have substituted \eqref{Optimal_PDF} and used variable transformation $t = 1+x$ and in step $n$, we use the exponential integral function $Ei(.)$ to solve the integral as presented in \cite{2014520}. This completes the proof of Lemma \ref{Lemma1}.
	
	\section{Proof of CDF under Hop-by-Hop Relay Selection}\label{app7:3}
	Let us define $\gamma_{ij}=\frac{P}{\sigma^2} |h[j,i,l]|^2$ as the received SNR of the relay link between the transmitting node $i$ and the receiving node $j$. Since, each link in the network follows a Rayleigh fading distribution with mean zero and variance $\sigma_1^2$, we can write the CDF of $\gamma_{ij}$ as,
	\begin{align}
		F_{\gamma_{ij}}(x) = 1-e^{(-x/\sigma_a^2)}, \; \forall i,j,
	\end{align}
	where $\sigma_a^2=\frac{P\sigma_1^2}{\sigma^2}$. Next, we proceed to analyze the maximum SNR in hop $l$ denoted by $\gamma_{max}^{(l)}$ under hop-by-hop relay selection strategy. 
	If we consider any hop other than the last hop, the forwarding link is selected as the link with the maximum received SNR in that hop. Since, all $M$ links in a given hop are independent from each other, the CDF of the maximum SNR in a given hop $l$ where $l \in \{1,..,L-1\}$ can be derived as,
	\begin{align}
		F_{\gamma_{max}^{(l)}}(x) = \bigg(1-e^{(-x/\sigma_a^2)}\bigg)^M.
	\end{align}
	However, in the last hop there is no choice to make as the destination is fixed. As such, 
	\begin{align} 
		F_{\gamma_{max}^{(L)}}(x) = 1-e^{(-x/\sigma_a^2)}.
	\end{align}
	The end-to-end SNR is taken as the minimum SNR over all $L$ hops. Since all these maximum SNRs are independent of each other, the CDF of $\gamma^{min}$ can be expressed as,
	\begin{align}
		F_{\gamma^{min}}^{\textrm{hop}}(x) &= 1-e^{(-x/\sigma_a^2)}\bigg(1-\bigg(1-e^{(-x/\sigma_a^2)}\bigg)^M\bigg)^{L-1}.
	\end{align}
	
	\section{Proof of CDF under Ad-hoc Relay Selection}\label{app7:4}
	Let us define $\gamma_{ij}=\frac{P}{\sigma^2} |h[j,i,l]|^2$ as the received SNR of the relay link between the transmitting node $i$ and the receiving node $j$. Since, each link in the network follows a Rayleigh fading distribution with mean zero and variance $\sigma_1^2$, we can write the CDF of $\gamma_{ij}$ as,
	\begin{align}
		F_{\gamma_{ij}}(x) = 1-e^{(-x/\sigma_a^2)}, \; \forall i,j,
	\end{align}
	where $\sigma_a^2=\frac{P\sigma_1^2}{\sigma^2}$. Next, we proceed to analyze the end-to-end SNR under ad-hoc relay selection strategy. 
	Under this relay selection strategy, the forwarding link is selected as the link with the maximum received SNR in hop $l$ where $l \in \{1,2,...,L-2\}$. Since each of $M$ links are independent from each other, the CDF of the maximum SNR in a given hop can be derived as,
	\begin{align}\label{eq_hop_76}
		F_{\gamma_{max}^{(l)}}(x) = \bigg(1-e^{(-x/\sigma_a^2)}\bigg)^M,
	\end{align}
	where $\gamma_{max}^{(l)}$ denotes the maximum SNR in hop $l$. However, the last two hops are combined such that the new effective SNR for a given link $j$ in hop $L-1$ is the minimum of two links. Therefore, the CDF of this new effective SNR in hop $L-1$ can be expressed as,
	\begin{align}
		F_{\gamma^{L-1}_{j1}} (x) =1-e^{(-2x/\sigma_a^2)}.
	\end{align}
	Since, the relay in hop $L-1$ is selected such that this new effective SNR is maximized, we can write the CDF of the maximum SNR in hop $L-1$ as, 
	\begin{align}\label{eq_adhoc_1}
		F_{\gamma_{max}^{(L-1)}}(x) &= \bigg(1- e^{(-2x/\sigma_a^2)}\bigg)^M.
	\end{align}
	The end-to-end SNR is taken as the minimum SNR over all $L$ hops. Since all these maximum SNRs are independent of each other, the CDF of $\gamma^{min}$ can be expressed as,
	\begin{align}
		F_{\gamma^{min}}^{\textrm{ad-hoc}}(x) &= 1-\biggl(1-\bigg(1- e^{(-2x/\sigma_a^2)}\bigg)^M\biggr) \bigg(1-\bigg(1-e^{(-x/\sigma_a^2)}\bigg)^M\bigg)^{L-2}.
	\end{align}	
	
	\section{Proof of CDF under Block-by-Block Relay Selection}\label{app7:5}
	Let us define $\gamma_{ij}=\frac{P}{\sigma^2} |h[j,i,l]|^2$ as the received SNR of the relay link between the transmitting node $i$ and the receiving node $j$. Since, each link in the network follows a Rayleigh fading distribution with mean zero and variance $\sigma_1^2$, we can write the CDF of $\gamma_{ij}$ as,
	\begin{align}
		F_{\gamma_{ij}}(x) = 1-e^{(-x/\sigma_a^2)}, \; \forall i,j,
	\end{align}
	where $\sigma_a^2=\frac{P\sigma_1^2}{\sigma^2}$. Next, we proceed to analyze the effective SNR of each block under block-by-block relay selection strategy. Let us divide $L$ hops into blocks of $w$ hops where the total number of blocks are given by $T=\left\lceil \frac{L}{w}\right\rceil$. Since we consider the block size $w=2$, the forwarding link is selected as the link with the maximum received SNR in two consecutive hops such that the signal can end up in any node at the end of the second hop for a given block $t$ where $t \in \{1,2,...,T-1\}$. As such, the maximum minimum SNR across relay $i$ in the first hop of the block can be written as, 
	\begin{align}
		\gamma_i^{(t)} = \textrm{min}\{\gamma_{1,i}^{2t-1},X_i^t\},
	\end{align}
	where $\gamma_{1,i}^{2t-1}$ denotes the SNR between transmitter $1$ and receiver $i$ in hop $2t-1$ and $X_i^t$ denotes the maximum possible SNR in the second hop given relay $i$ in first hop and can be expressed as,
	\begin{align}
		X_i^t =  {\underset{j \in \{1,...,M\}} {\textrm{max} }} \{\gamma_{i,j}^{2t}\}.
	\end{align}
	Since each of these links are independent from each other, the CDF of the maximum minimum SNR across relay $i$ in the first hop of block $t$ can be derived as,
	\begin{align}
		F_{\gamma_{i,max}^{(t)}}(x) &=   1-e^{(-x/\sigma_a^2)}\bigg(1-(1-e^{(-x/\sigma_a^2)})^M\bigg).
	\end{align}
	%	\newpage
	Then we have $M$ independent paths across $M$ relays in the first hop of the block and as such the effective SNR of a given block $t$ where $t \in \{1,2,...,T-1\}$ can be written as,
	\begin{align}
		F_{\gamma_{max}^{(t)}}(x) &= \bigg(1-e^{(-x/\sigma_a^2)}\bigg(1-(1-e^{(-x/\sigma_a^2)})^M\bigg)\bigg)^M.
	\end{align}
	For the final block, the effective SNR is computed similar to the ad-hoc relay selection approach given in \eqref{eq_adhoc_1} as, 
	\begin{align} 
		F_{\gamma_{max}^{(T)}}(x) 
		&= \bigg(1- e^{(-2x/\sigma_a^2)}\bigg)^M.
	\end{align}
	Then, the end-to-end SNR can be derived as the minimum of these effective SNRs over all the blocks. Since, all $T$ blocks are independent of each other, the CDF of $\gamma^{min}$ can be expressed as,
	\begin{align}
		F_{\gamma^{min}}^{\textrm{block}}(x) &\!=\! 1\!-\!\biggl(1\!-\!\bigg(1\!-\! e^{(-2x/\sigma_a^2)}\bigg)^M\biggr)\bigg(1\!-\!\bigg(1\!-\!e^{(-x/\sigma_a^2)}\bigg(1\!-\!(1\!-\!e^{(-x/\sigma_a^2)})^M\bigg)\bigg)^M\bigg)^{T\!-\!1}.
	\end{align}
	
	\section{Proof of CDF under Sliding Window based Relay Selection}\label{app7:6}
	Let us define $\gamma_{ij}=\frac{P}{\sigma^2} |h[j,i,l]|^2$ as the received SNR of the relay link between the transmitting node $i$ and the receiving node $j$. Since, each link in the network follows a Rayleigh fading distribution with mean zero and variance $\sigma_1^2$, we can write the CDF of $\gamma_{ij}$ as,
	\begin{align}
		F_{\gamma_{ij}}(x) = 1-e^{(-x/\sigma_a^2)}, \; \forall i,j,
	\end{align}
	where $\sigma_a^2=\frac{P\sigma_1^2}{\sigma^2}$. Next, we proceed to analyze the end-to-end SNR under sliding window based relay selection strategy. For a window size of $w$, we have $T=L-w+1$ windows when the window slides by one hop at a time. Therefore, for a given window $t$ where $t \in \{1,...,T-1\}$, we select the best path similar to block-by-block relay selection with $w=2$. Let us define $X_i^t$ to represent the maximum possible SNR across relay $i$ in the first hop within the window $t$ as,
	\begin{align}
		X_i^t =  {\underset{j \in \{1,...,M\}} {\textrm{max} }} \{\gamma_{i,j}^{t+1}\},
	\end{align}
	where $\gamma_{i,j}^{t+1}$ denotes the SNR between transmitter $i$ and receiver $j$ in hop $t+1$. Since, all $M$ links in the second hop are independent with respect to each other, the CDF and the PDF of the random variable $X_i^t$ can be expressed as,
	\begin{align}
		F_{X_i^t} (x) &= (1-e^{-x/\sigma_a^2})^M, \\
		f_{X_i^t} (x) &= M\bigg(\dfrac{1}{\sigma_a^2}\bigg)e^{-x/\sigma_a^2}(1-e^{-x/\sigma_a^2})^{M-1},
	\end{align}
	respectively. Next, let us define the variable $\gamma_{max}^{(t)}$ to represent the SNR between the transmitter and the selected relay of the first hop such that the effective SNR is maximized for the window $t$. Therefore, the conditional CDF of $\gamma_{max}^{(t)}$ can be written as,
	\begin{align}
		F_{\gamma_{max}^{(t)}}(x) &= \sum_{k=1}^{M}I(\gamma_{1k}^t \le x)P(\gamma_{max}^{(t)}=\gamma_{1k}^t) = \left\{
		\begin{array}{ll}
			1 \qquad & \gamma_{1k}^t \le x, \; \forall k \\
			\sum_{k=1}^{N}P(\gamma_{max}^{(t)}\!=\!\gamma_{1k}^t) \qquad & n[\gamma_{1k}^t \le x] \!=\! N, \\ & n[x < \gamma_{1k}^t] \!=\! M\!-\!N\\
			0 & x < \gamma_{1k}^t, \; \forall k,\\
		\end{array}
		\right.
	\end{align}
	where $n[\gamma_{1k}^t \le x]$ represents the number of links which has $\gamma_{1k}^t \le x$. If we consider the scenario $n[\gamma_{1k}^t \le x] = N, n[x < \gamma_{1k}^t] = M-N$, then there are $P^M_N$ possible sub-regions that satisfies this as we are only interested in the $P(\gamma_{max}^{(t)}=\gamma_{1k}^t)$ for $\gamma_{1k}^t \le x$, when $P^M_N$ denotes the number of possible permutations when $N$ units are selected from $M$.
	Therefore, the unconditional CDF of $\gamma_{max}^{(t)}$ can be derived as,  
	\begin{align}
		F_{\gamma_{max}^{(t)}}(x) =& \int_{\gamma_{11}^t =0}^{\infty}\dots\int_{\gamma_{1M}^t =0}^{\infty} P(\gamma_{max}^{(t)} \le x) f_{\gamma_{11}^t}(\gamma_{11}^t)\dots f_{\gamma_{1M}^t}(\gamma_{1M}^t) \, d\gamma_{11}^t \dots d\gamma_{1M}^t  \nonumber \\
		=& (1-e^{-x/\sigma_a^2})^M \!+\!\! \sum_{N=1}^{M-1} P^M_N e^{-(M\!-\!N)x} \bigg(\int_{\gamma_{11}^t =0}^{x} \int_{\gamma_{12}^t =\gamma_{11}^t}^{x}\!\!\dots\int_{\gamma_{1N}^t =\gamma_{1,N-1}^t}^{x} \sum_{k=1}^{N}P(\gamma_{max}^{(t)}\!=\!\gamma_{1k}^t) \nonumber \\ & \hspace{200pt}  e^{-\gamma_{11}^t} \dots e^{-\gamma_{1N}^t} d\gamma_{11}^t \dots d\gamma_{1N}^t \bigg), \label{eq_58}
	\end{align}
	where $\sum_{k=1}^{N}P(\gamma_{max}^{(t)}=\gamma_{1k}^t)$  is given by the following lemma.	
	\begin{lemma}\label{lemma_D1}
		For a network with $M$ independent links and the CDF of any link given as $F_{\gamma_{1k}^t} (x) = 1-e^{-x/\sigma_a^2}$, the probability $\sum_{k=1}^{N}P(\gamma_{max}^{(t)}=\gamma_{1k}^t)$ for the region $\gamma_{11}^t < \gamma_{12}^t < ... < \gamma_{1N}^t \le x < \gamma_{1,N+1}^t,...,\gamma_{1M}^t$ is given by 
		\begin{align*}
			\sum_{k=1}^{N}P(\gamma_{max}^{(t)}\!\!=\!\gamma_{1k}^t) \!=\!\!& \sum_{k=1}^{N} \bigg(\!\dfrac{M\!-\!N}{M\!-\!k}\!\bigg) (1-e^{-\gamma_{1k}^t/\sigma_a^2})^{M(M-k)}  \!-\! \bigg(\!\dfrac{M\!-\!N}{M\!-\!k\!+\!1}\!\bigg) (1-e^{-\gamma_{1k}^t/\sigma_a^2})^{M(M\!-\!k\!+\!1)}.
		\end{align*}
	\end{lemma} 
	
	\begin{proof}	
		We start the proof by noting that the following equality can be easily proved with the method of induction
		\begin{align} \label{eq_sliding_2}
			\sum_{n_s \!=\! 0}^{k\!-\!1}\!\!\! \sum_{{\underset{\mathbf{S} \subset \{1,\dots,K\!-\!1\}} {n[S] = n_s}}}\!\!\! \prod_{i\in\mathbf{S}}^{k\!-\!1}\! [1\!-\!(1\!-\!e^{-\gamma_{1i}^t/\sigma_a^2})^M]\!\prod_{l\not\in\mathbf{S}}^{k\!-\!1}\! (1\!-\!e^{-\gamma_{1l}^t/\sigma_a^2})^M \!\!=\! 1.
		\end{align} 
		Next, we define two new probabilities $P_{1k}^t, P_{2k}^t$ and write
		\begin{align*}
			P(\gamma_{max}^{(t)}\!=\!\gamma_{1k}^t) &\!=\!P(\textrm{min}\{\gamma_{1k}^t,X_k^t\} \! \ge \! \textrm{ min} \{\gamma_{1i}^t,X_i^t\},\; \forall i \neq k) \nonumber \\
			&\!=\! \left\{
			\begin{array}{ll}
				\!\!\! P_{1k}^t \qquad & \textrm{ if } X_k^t \ge \gamma_{1k}^t, X_i^t \ge  \gamma_{1i}^t, X_l^t < \gamma_{1l}^t, \gamma_{1k}^t \ge \gamma_{1i}^t,X_l^t, \; \forall i \in \mathbf{S}, l \not\in \mathbf{S} \cup k\\ \\
				\!\!\! P_{2k}^t \qquad & \textrm{ if } X_k^t < \gamma_{1k}^t, X_i^t \ge  \gamma_{1i}^t, X_l^t <  \gamma_{1l}^t,  X_k^t > \gamma_{1i}^t,X_l^t, \; \forall i \in \mathbf{S}, l \not\in \mathbf{S} \cup k
			\end{array}
			\right.
		\end{align*}
		We have $\gamma_{1i}^t > \gamma_{1k}^t, \; \forall i > k$ in the region $\gamma_{11}^t < ... < \gamma_{1N}^t \le x < \gamma_{1,N+1}^t,...,\gamma_{1M}^t$. Therefore, if there exists $i \in \mathbf{S}$ such that $i >k$ then we have $I(\gamma_{1k}^t \ge \gamma_{1i}^t) = 0$ and $P(\gamma_{1k}^t > X_k^t \ge \gamma_{1i}^t)=0$ which makes $P_{1k}^t, P_{2k}^t = 0$. Therefore, $\mathbf{S} \subset \{1,...,k-1\}$ and the cardinality of set $\mathbf{S}$ given by $n[\mathbf{S}] \le k-1$. Based on that we write $P_{1k}^t$ as, 
		\begin{align}
			P_{1k}^t &= \sum_{n_s = 0}^{k-1}\!\sum_{{\underset{\mathbf{S} \subset \{1,...,K-1\}} {n[S] = n_s}}} P(X_k^t \ge \gamma_{1k}^t)\prod_{i\in\mathbf{S}}^{k-1} I(\gamma_{1k}^t \ge \gamma_{1i}^t) P( X_i^t \ge  \gamma_{1i}^t)\!\prod_{l\not\in\mathbf{S}\cup\{k\}}^{M} P(X_l^t < \min\{\gamma_{1k}^t,\gamma_{1l}^t\}) \nonumber \\ 
			&=\sum_{n_s = 0}^{k-1}\!\!\sum_{{\underset{\mathbf{S} \subset \{1,...,K-1\}} {n[S] = n_s}}}\!\!\!\!\!\![1\!-\!(1\!-\!e^{-\gamma_{1k}^t/\sigma_a^2})^M]\!\prod_{i\in\mathbf{S}}^{k-1} [1\!-\!(1\!-\!e^{-\gamma_{1i}^t/\sigma_a^2})^M]\!\prod_{l\not\in\mathbf{S}}^{k-1} (1\!-\!e^{-\gamma_{1l}^t/\sigma_a^2})^M  \!\!\!\prod_{l=k+1}^{M}\!\! (1\!-\!e^{-\gamma_{1k}^t/\sigma_a^2})^M \nonumber \\ 
			&\!=\![1\!-\!(1\!-\!e^{-\gamma_{1k}^t/\sigma_a^2})^M](1\!-\!e^{-\gamma_{1k}^t/\sigma_a^2})^{M(M\!-\!k)}\!\!\!\sum_{n_s = 0}^{k-1}\!\!\! \sum_{{\underset{\mathbf{S} \subset \{1,...,K-1\}} {n[S] = n_s}}}\!\!\!\prod_{i\in\mathbf{S}}^{k-1}\! [1\!-\!(1\!-\!e^{-\gamma_{1i}^t/\sigma_a^2})^M]\!\prod_{l\not\in\mathbf{S}}^{k-1}\! (1\!-\!e^{-\gamma_{1l}^t/\sigma_a^2})^M \nonumber \\ 
			&\!= [1-(1-e^{-\gamma_{1k}^t/\sigma_a^2})^M] (1-e^{-\gamma_{1k}^t/\sigma_a^2})^{M(M-k)}. \label{eq_60}
		\end{align}
		where the last line follows from \eqref{eq_sliding_2}. 
		
		Similarly, we can write $P_{2k}^t$ as, 
		\begin{align}
			P_{2k}^t &= \sum_{n_s = 0}^{k-1}\!\sum_{{\underset{\mathbf{S} \subset \{1,...,K-1\}} {n[S] = n_s}}}\!\prod_{i\in\mathbf{S}}^{k-1} P(\gamma_{1i}^t \le X_k^t < \gamma_{1k}^t)P(X_i^t \ge \gamma_{1i}^t) \prod_{l\not\in\mathbf{S}\cup\{k\}}^{M}\! P(X_l^t \!<\! \min\{X_k^t,\gamma_{1l}^t\},X_k^t \!<\! \gamma_{1k}^t). \label{eq_61}
		\end{align}
		Let us define $x$ where $x \in X_k^t$ such that $x \le \gamma_{1k}^t$ and write the conditional probability as,
		\begin{align}			
			P_{2k}^t|x &=\sum_{n_s = 0}^{k-1} \sum_{{\underset{\mathbf{S} \subset \{1,...,K-1\}} {n[S] = n_s}}}\prod_{i\in\mathbf{S}}^{k-1} I(\gamma_{1i}^t \le x)P(X_i^t \ge \gamma_{1i}^t) \prod_{l\not\in\mathbf{S}\cup\{k\}}^{M} P(X_l^t < \min\{x,\gamma_{1l}^t\}) . \label{eq_62}
		\end{align}
		Then we integrate $P_{2k}^t|x$ with respect to $x$ to get the unconditional probability given in 
		\begin{align}
			P_{2k}^t &= \int_{0}^{\gamma_{1k}^t} P_{2k}^t|x \; f_{X_k^t} (x) dx = \int_{0}^{\gamma_{11}^t} P_{2k}^t|x \; f_{X_k^t} (x) dx + \sum_{j=2}^{k} \int_{\gamma_{1,j-1}^t}^{\gamma_{1j}^t} \! \! P_{2k}^t|x \; f_{X_k^t} (x) dx  \nonumber \\
			&= \int_{0}^{\gamma_{11}^t} \prod_{l\neq k}^{M} P(X_l^t < x) f_{X_k^t} (x)\; dx \; + \sum_{j=2}^{k} \int_{\gamma_{1,j-1}^t}^{\gamma_{1j}^t}\biggl[ \sum_{n_s = 0}^{k-1} \sum_{{\underset{\mathbf{S} \subset \{1,...,K-1\}} {n[S] = n_s}}}\prod_{i\in\mathbf{S}}^{k-1} I(\gamma_{1i}^t \le x)P(X_i^t \ge \gamma_{1i}^t) \nonumber \\
			& \hspace{200pt} \prod_{l\not\in\mathbf{S}}^{j-1} P(X_l^t < \gamma_{1l}^t) \prod_{l=j,l\not\in\mathbf{S}\cup\{k\}}^{M} P(X_l^t < x) \biggr]f_{X_k^t} (x) dx. \label{eq_63} 
		\end{align}
		Note that when $\gamma_{1,j-1}^t \le x \le \gamma_{1j}^t$, if there exist $i\in\mathbf{S}$ such that $i\ge j$ then $I(\gamma_{1i}^t<x)=0$ thus the term relevant for that realization of $\mathbf{S}$ becomes zero. Therefore, we can consider $\mathbf{S} \subset \{1,...,j-1\}$ with $n[\mathbf{S}]\le j-1$ instead of $\mathbf{S} \subset \{1,...,k-1\}$ with $n[\mathbf{S}]\le k-1$ and simplify $P_{2k}^t$ ,  % write
		\begin{align}
			P_{2k}^t &= \int_{0}^{\gamma_{11}^t} \dfrac{M}{\sigma_a^2}e^{-x/\sigma_a^2}(1-e^{-x/\sigma_a^2})^{M^2-1}\; dx \; + \sum_{j=2}^{k} \int_{\gamma_{1,j-1}^t}^{\gamma_{1j}^t}  \dfrac{M}{\sigma_a^2}e^{-x/\sigma_a^2}(1-e^{-x/\sigma_a^2})^{M(M-j+1)-1}\; dx  \nonumber \\ & \hspace{150pt} \bigg(\sum_{n_s = 0}^{j-1} \sum_{{\underset{\mathbf{S} \subset \{1,...,j-1\}} {n[S] = n_s}}} \prod_{i\in\mathbf{S}}^{j-1}[1-(1-e^{-\gamma_{1i}^t/\sigma_a^2})^M] \prod_{l\not\in\mathbf{S}}^{j-1}  (1-e^{-\gamma_{1l}^t/\sigma_a^2})^M \bigg) \nonumber \\ 
			&=\! \int_{0}^{\gamma_{11}^t}\! \dfrac{M}{\sigma_a^2}e^{-x/\sigma_a^2}(1\!-\!e^{-x/\sigma_a^2})^{M^2-1}\; dx \!+\! \sum_{j=2}^{k} \int_{\gamma_{1,j-1}^t}^{\gamma_{1j}^t}\!\!\!  \dfrac{M}{\sigma_a^2}e^{-x/\sigma_a^2}(1\!-\!e^{-x/\sigma_a^2})^{M(M-j+1)-1}\; dx  \nonumber \\
			&=\! \dfrac{1}{M}(1\!-\!e^{-\gamma_{11}^t/\sigma_a^2})^{M^2} \!+\! \sum_{j=2}^{k} \dfrac{1}{M\!-\!j\!+\!1}\bigg[ (1\!-\!e^{-\gamma_{1j}^t/\sigma_a^2})^{M(M-j+1)}-(1-e^{-\gamma_{1,j-1}^t/\sigma_a^2})^{M(M-j+1)}\bigg], \label{eq_64} 
		\end{align}
		where the second last line follows from \eqref{eq_sliding_2}. Next, using the method of induction we write 
		\begin{align} 
			\sum_{k=1}^{N} P(\gamma_{max}^{(t)}&=\gamma_{1k}^t) = \sum_{k=1}^{N} P_{1k}^t+P_{2k}^t  \nonumber \\
			& \qquad \quad = \sum_{k=1}^{N}\biggl[ \bigg(\dfrac{M-k}{M-k}(1-e^{-\gamma_{1k}^t/\sigma_a^2})^{M(M-k)}\!-\!\dfrac{M-k}{M-k+1}(1-e^{-\gamma_{1k}^t/\sigma_a^2})^{M(M-k+1)}\bigg) \nonumber \\ & \hspace{50pt} + \sum_{j=1}^{k-1} \dfrac{1}{M-j+1}(1-e^{-\gamma_{1j}^t/\sigma_a^2})^{M(M-j+1)}-\dfrac{1}{M-j}(1-e^{-\gamma_{1j}^t/\sigma_a^2})^{M(M-j)}\biggr]\nonumber \\
			&\qquad \quad = \sum_{k=1}^{N}\!\bigg(\!\dfrac{M\!-\!N}{M\!-\!k}\!\bigg)(1-e^{-\gamma_{1k}^t/\sigma_a^2})^{M(M-k)}\!-\!\bigg(\!\dfrac{M\!-\!N}{M\!-\!k\!+\!1}\!\bigg) (1-e^{-\gamma_{1k}^t/\sigma_a^2})^{M(M-k+1)}. \label{eq_65}
		\end{align}
		which completes the proof of Lemma \ref{lemma_D1}. 
	\end{proof}
	
	For the final window that consists of the last two hops, the relay selection will be fixed for both hops and the effective SNR is computed similar to \eqref{eq_adhoc_1} as,	
	\begin{align} \label{eq_hop_716}
		F_{\gamma_{max}^{(T)}}(x) 
		&= \bigg(1- e^{(-2x/\sigma_a^2)}\bigg)^M.
	\end{align}
	Then the end-to-end SNR can be derived as the minimum of these effective SNRs over all $T$ windows. Under the independent window model, those effective SNRs are independent of each other. Therefore, we have $T$ independent windows. As such,  we derive the CDF of $\gamma^{min}$ as ,  
	\begin{align}
		F_{\gamma^{min}}^{\textrm{sliding}}(x) &= 1-\biggl(1-\bigg(1- e^{(-2x/\sigma_a^2)}\bigg)^M\biggr)\biggl(1-(1-e^{-x/\sigma_a^2})^M - \sum_{N=1}^{M-1} P^M_N e^{-(M-N)x/\sigma_a^2} \nonumber \\ & \;\bigg(\int_{\gamma_{11}^t =0}^{x}\int_{\gamma_{12}^t =\gamma_{11}^t}^{x}...\int_{\gamma_{1N}^t =\gamma_{1,N-1}^t}^{x}\sum_{k=1}^{N}P(\gamma_{max}^{(t)}=\gamma_{1k}^t)\; e^{-\gamma_{11}^t}... e^{-\gamma_{1N}^t} d\gamma_{11}^t ... d\gamma_{1N}^t \bigg)\!\biggr)^{T\!-\!1}. \label{eq_hop_717}
	\end{align}
	where $\sum_{k=1}^{N}P(\gamma_{max}^{(t)}=\gamma_{1k}^t)$ for the region $\gamma_{11}^t < \gamma_{12}^t < ... < \gamma_{1N}^t \le x$ is given by Lemma \ref{lemma_D1}.

\end{document}